%% file: mbz-families_v3_quantum.tex
\documentclass[a4paper, 10pt, onecolumn, accepted=2025-07-22]{quantumarticle}
\pdfoutput=1
\usepackage{tikz,tikz-cd}
\usepackage{pgfplots}\pgfplotsset{compat=1.18}

\input{common_preamble_quantum}

\title{Multipartite Embezzlement of Entanglement}

\author{Lauritz van Luijk\,\orcidlink{0000-0003-3153-549X}, Alexander Stottmeister\,\orcidlink{0000-0002-0145-0877}, Henrik Wilming\,\orcidlink{0000-0002-0306-7679}}
\affiliation{Leibniz Universit\"at Hannover, Institut f\"ur Theoretische Physik,  \\ Appelstraße 2, 30167 Hannover, Germany}
\date{July 28th, 2025}

\begin{document}

\maketitle

\begin{abstract}
Embezzlement of entanglement refers to the task of extracting entanglement from an entanglement resource via local operations and without communication while perturbing the resource arbitrarily little. Recently, the existence of embezzling states of bipartite systems of type III von Neumann algebras was shown.
However, both the multipartite case and the precise relation between embezzling states and the notion of embezzling families, as originally defined by van Dam and Hayden, were left open. Here, we show that finite-dimensional approximations of multipartite embezzling states form multipartite embezzling families.
In contrast, not every embezzling family converges to an embezzling state.
We identify an additional consistency condition that ensures that an embezzling family converges to an embezzling state.
This criterion distinguishes the embezzling family of van Dam and Hayden from that of Leung, Toner, and Watrous.
The latter generalizes to the multipartite setting.
By taking a limit, we obtain a multipartite system of commuting type III$_1$ factors on which every state is an embezzling state.
We discuss our results in the context of quantum field theory and quantum many-body physics.
As open problems, we ask whether vacua of relativistic quantum fields in more than two spacetime dimensions are multipartite embezzling states and whether multipartite embezzlement allows for an operator-algebraic characterization.
\end{abstract}
\tableofcontents
\null

\section{Introduction}

Entanglement is the property of quantum states that is impossible to create with local operations.
However, it turns out that entanglement can be \emph{embezzled} using local operations: In \cite{van_dam2003universal}, van Dam and Hayden discovered that there exist families of finite dimensional, pure, bipartite entangled states $\Omega_n$ (on systems $A,B$ and with Schmidt rank $n$) shared between Alice and Bob such that any finite-dimensional, pure, bipartite, entangled state $\Psi_{A'B'}\in\CC^d\ox\CC^d$ may be extracted from them while perturbing the original state arbitrarily little:
\begin{align}
    u_{AA'}u_{BB'} \,\Omega_{n,AB}\ox \ket 1_{A'}\ket 1_{B'} \approx_\eps \Omega_{n,AB} \ox \Psi_{A'B'}
\end{align}
with $\eps\to 0$ as $n\to \oo$. Here, $u_{AA'},u_{BB'}$ are suitably local unitary operators acting on the indicated subsystems. 
In other words, Alice and Bob can transform the product state $\ket 1_{A'}\ket 1_{B'}$ into an arbitrary entangled state, while perturbing the \emph{embezzlement resource} $\Omega_{n,AB}$ arbitrarily little. 
Equivalently, for any pair of bipartite entangled pure states $\Psi_{A'B'},\Phi_{A'B'}\in\CC^d\otimes\CC^d$ and every $\eps>0$ there exist local unitaries $u_{AA'},u_{BB'}$ and a sufficiently large $n$ such that
\begin{align}
    u_{AA'}u_{BB'} \,\Omega_{n,AB}\ox \Phi_{A'B'} \approx_\eps \Omega_{n,AB} \ox \Psi_{A'B'}.
\end{align}
In a sense, the embezzlement resource state acts like a catalyst, see \cite{datta_catalysis_2023,lipka-bartosik_catalysis_2024} for reviews on catalysis in quantum information theory.  

Embezzlement has found numerous applications in quantum information theory \cite{Bennett_2014,berta_quantum_2011,Leung2019,coladangelo_two-player_2020,Leung2013coherent,regev_quantum_2013,dinur_parallel_2015,cleve_perfect_2017,coladangelo_two-player_2020}.
General properties of embezzlement families such as the one of van Dam and Hayden have been studied in detail in \cite{leung_characteristics_2014,zanoni_complete_2024}.

A family of states that can be used for embezzlement of entanglement that is quite distinct from that used by van Dam and Hayden has been constructed by Leung, Toner and Watrous (LTW) \cite{Leung2013coherent}. The LTW embezzling family is also known as `universal construction' of embezzling families, as variants of it allow to construct embezzling families for large classes of resource theories, see \cite{lipka-bartosik_catalysis_2024} for a comprehensive discussion. 
However, it is less efficient than the van Dam-Hayden (vDH) family because the required Hilbert space dimension grows much quicker as the error decreases. 
A first natural question arises: 
\begin{enumerate}
    \item Is there a fundamental difference between the vDH family and the LTW family?
\end{enumerate}

It was recently shown that in bipartite systems with infinitely many degrees of freedom, modeled by two commuting von Neumann algebras $\M_A$ and $\M_B=\M_A'$\footnote{If $\M$ is a von Neumann algebra on $\H$, its commutant $\M'$ is defined as $\M' = \{x\in \B(\H)\ :\ [x,a]=0,\ \forall a\in \M\}.$} acting on a joint Hilbert space $\H$, it can happen that \emph{every single density matrix} $\rho$ on the bipartite Hilbert space $\H$ is embezzling \cite{long_paper,short_paper}: For every $d\in\NN$, every two pure bipartite quantum state $\Psi,\Phi \in \CC^d\otimes \CC^d$, and every $\eps>0$, there exist local unitaries $u_A\in \M_A\otimes M_d(\CC)$ and $u_B\in \M_B\otimes  M_d(\CC)$ such that
\begin{align}
    \norm{u_A u_B(\rho \otimes |\Phi\rangle\!\langle \Phi|)u_A^* u_B^* - \rho\otimes |\Psi\rangle\!\langle\Psi|}_1 < \eps.
\end{align}
We call any state $\rho$ on $\H$ for which the above holds an \emph{embezzling state}. Bipartite quantum systems on which all density matrices are embezzling states are called \emph{universal embezzlers}. Whether embezzling states exist is determined by the type classification of the von Neumann algebra $\M_A$ (and $\M_B$) \cite{long_paper,short_paper}. In particular, only type III von Neumann algebras can host embezzling states, and bipartite systems are universal embezzlers if and only if $\M_A$ is type III$_1$. 
A natural second question arises:
\begin{enumerate}[resume] 
\item What is the relation between embezzling states and embezzling families?
\end{enumerate}
In particular, it is natural to wonder whether embezzling states can, in any meaningful sense, arise as limits of embezzling families.

Finally, since the LTW family naturally generalizes to the multipartite setting, a third question immediately suggests itself:
\begin{enumerate}[resume]
    \item Are there multipartite quantum systems with infinitely many degrees of freedom which host \emph{multipartite embezzling states}?
\end{enumerate}
The purpose of this paper is to answer these questions and to clarify the relation between embezzling families and embezzling states.

\paragraph{Acknowledgments.} We thank Matthias Christandl and Reinhard F.\ Werner for interesting discussions. AS and LvL have been funded by a Stay Inspired Grant of the MWK Lower Saxony (Grant ID: 15-76251-2-Stay-9/22-16583/2022). The publication of this article was funded by the Open Access Fund of the Leibniz Universität Hannover.

\paragraph{Declarations.}
The authors have no conflicts of interest to declare.
No data was generated or processed in this work.

\paragraph{Notation and conventions.}
Throughout, we work in the setting of von Neumann algebras on separable Hilbert spaces. A von Neumann algebra $\M$ on a Hilbert space $\H$ is a weakly closed, unital *-algebra of bounded operators on $\H$. We typically suppress the Hilbert space. $\M$ is a factor if it has a trivial center: $\M \cap \M' = \CC 1$. 
We also adopt the convention that all tensor products between von Neumann algebras refer to the spatial tensor product, denoted by $\otimes$.
In the following, all operator algebras are assumed to be unital, and inclusions of operator algebras $\M\subset\N$ are unital unless specified otherwise.
We denote the scalar product of Hilbert space vectors $\Psi,\Phi\in\H$ as $\langle\Psi|\Phi\rangle$ and denote by $\{\ket{j}\}_{j=1}^{d}$ the canonical basis vectors on $\CC^d$.
For an operator $a$ on $\H$ and a vectors $\Psi,\Phi\in\H$ we sometimes write $\bra\Psi a\ket\Phi$ for $\braket{\Psi}{a\Phi}$.

\section{Overview of results and discussion}

In this section, we survey and discuss our main results. 
Throughout this paper, we model quantum systems using von Neumann algebras: Observables are described by the self-adjoint part of a von Neumann algebra $\M$. Quantum states are given by normal states on $\M$, i.e., ultraweakly continuous, positive, linear functionals $\omega:\M\to\CC$ with $\omega(1)=1$.
The expectation value of an observable $a$ in a state $\omega$ is given by $\omega(a)$.

We begin by setting up a proper framework to describe $N$-partite systems with infinitely many degrees of freedom: For $x\in[N]:=\{1,\ldots,N\}$, we consider collections of $N$ commuting von Neumann algebras $\M_x$ generating a von Neumann algebra $\M_{[N]} = \vee_x \M_x$.\footnote{$\M\vee\N$ denotes the von Neumann algebra generated by $\M$ and $\N$.}
This framework encompasses diverse situations, such as many-body systems in the thermodynamic limit that are partitioned into multiple infinite parts or relativistic quantum field theories with spacetime partitioned into several causally closed parts. 

A finite-dimensional $N$-partite quantum system described by the Hilbert space $(\CC^d)^{\otimes N}$ is simply given by $\M_{[N]} = \otimes_{x=1}^N M_d(\CC)$, i.e., each subsystem is described by one copy of $M_d(\CC)\subset \M_{[N]}$. 
Now consider a general $N$-partite system $(\M_x)_{x\in[N]}$ and a normal state $\omega$ on $\M_{[N]}$. We say that $\omega$ is a multipartite embezzling state if for every two pure states $\psi$ and $\phi$ on the $N$-partite system on $(\CC^d)^{\otimes N}$ and every $\eps>0$ there exist unitaries $u_x \in \M_x\ox M_d(\CC)$ such that
\begin{align}
    \norm{u(\omega\ox \phi)u^* - \omega\ox \psi} < \eps,\qquad u = \prod_x u_x.
\end{align}
Whenever the local von Neumann algebras $\M_x$ are generated by a sequence of increasing von Neumann algebras $\M^{(n)}_x$ (e.g., finite subchains of increasing length of an infinite spin-chain), we obtain a sequence of multipartite systems $(\M^{(n)}_x)_{x\in[N]}$ and associated marginal states $\omega^{(n)}$. 
Our first main result is that in such a situation, every embezzling state induces embezzling families:
\begin{introtheorem}[informal, see cp.~\cref{thm:state-to-family}]\label{introthm:state-to-family} 
    If $\omega$ is an embezzling state as above, the family of multipartite states $(\omega^{(n)})_n$ on $(\M_x^{(n)})_{x\in[N]}$ is a \emph{consistent embezzling family}: For all $d\in\NN$ and $\eps>0$ there exists $n=n(d,\eps)$, such that for every two pure states $\psi$ and $\phi$ on the $N$-partite system on $(\CC^d)^{\otimes N}$ there exist unitaries $u_x \in \M_x^{(n)}\ox M_d(\CC)$ such that
    \begin{align}
       \norm{u(\omega^{(m)}\ox \phi)u^* - \omega^{(m)}\ox \psi} < \eps,\qquad u = \prod_x u_x,\quad m\geq n. 
    \end{align}
\end{introtheorem}

Importantly, this result states that an embezzling state induces an embezzling family in a particularly strong sense: The unitaries necessary to achieve a given error $\eps$ and dimension $d$ for the target states only require a minimal `size' of the resource and remain valid if the latter is enlarged. It is this property that we refer to as \emph{(quasi-)consistency} (see \cref{def:consistent} for the distinction).

A direct corollary of \cref{introthm:state-to-family} is as follows: Consider a bipartite embezzling ground state of a many-body system in the thermodynamic limit or the vacuum of a relativistic quantum field theory (these exist, see, for example, \cite{long_paper,van_luijk_critical_2024}). 
A priori, it would seem that to embezzle a state, it may be necessary for both agents to act on the entirety of their infinite subsystems. The theorem implies that, for any fixed precision, the agents, in fact, only have to act on a finite subsystem. Moreover, we can even show that the agents don't have to access the infinite-volume ground state or the exact vacuum to exploit the embezzlement phenomenon: A sufficiently good approximation, e.g., via finite-size ground states \cite{van_luijk_critical_2024} or scaling limits \cite{morinelli2021scaling_limits_wavelets, osborne2023cft_lattice_fermions}, can serve as an embezzling state for a given error $\eps$ and dimensionality $d$ of target states (see \cref{prop:mbzfam_quasi}).
\begin{figure}
\centering
\includegraphics{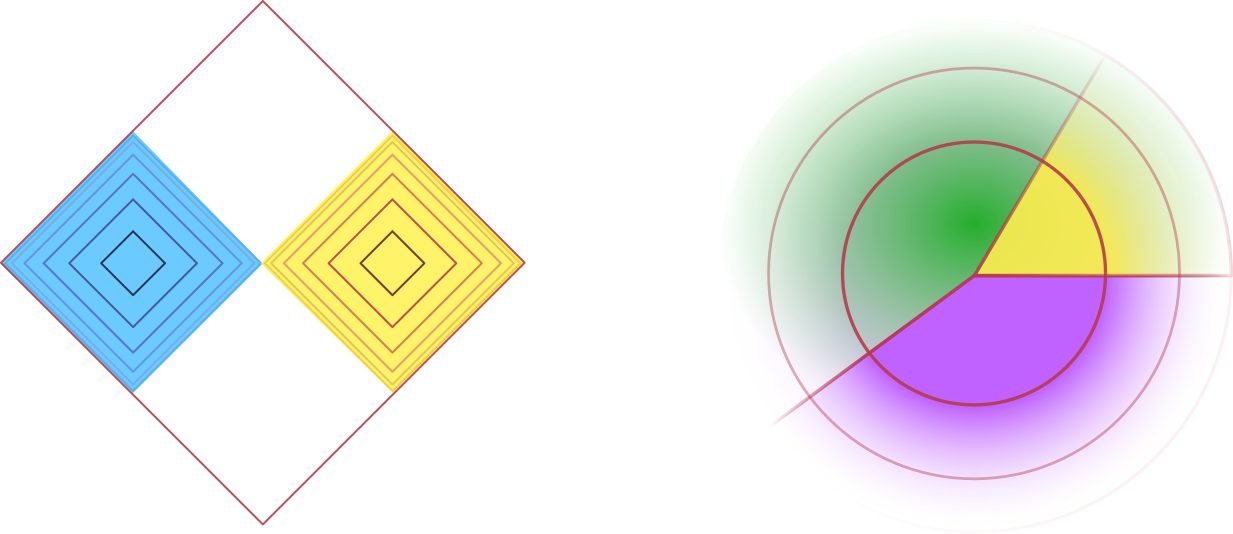}
\caption{Application of \cref{introthm:state-to-family}. \emph{Left:} The vacuum sector of a relativistic quantum field theory on Minkowski space, depicted by a Penrose diagram, is considered as bipartite system with respect to the left and right Rindler wedges $\mathcal W_L$ and $\mathcal W_R$. For each wedge $\mathcal W_i$, we can choose an increasing sequence $\mathcal C\up{n}_i\subset \mathcal W_i$ of double cones that asymptotically exhausts the wedge. Since the vacuum state $\Omega$ is a bipartite embezzler, its restrictions $\omega\up n$ to the subsystems $\mathcal C\up{n}_L\cup \mathcal C\up{n}_R$ are a bipartite embezzling family. \emph{Right:} A piece-of-cake subdivision of the plane into cones (colors). If the ground state of a many-body system on a 2D lattice is a tripartite embezzler with respect to the subdivision, one obtains a tripartite embezzling family by intersecting each cone with balls of increasing radius. The same would be true for quantum field theories. Our main open question is whether such multipartite embezzling states arise in many-body physics or quantum field theory.}
\label{fig:illustration}
\end{figure}

We can also ask and answer the converse question: Suppose we have a consistent embezzling family on $\omega^{(n)}$ on multipartite systems $(\M^{(n)}_x)_{x\in[N]}$, meaning that $\omega^{(m)} = \omega^{(n)}\restriction \M^{(m)}_{[N]}$ if $m\leq n$. Does this imply that a multipartite embezzling state exists as a limiting object?

\begin{introtheorem}[informal, cp.~\cref{thm:consistent-to-state}]\label{introthm:family-to-state} Let $\omega^{(n)}$ be a consistent embezzling family on an increasing sequence of $N$-partite systems $(\M^{(n)}_x)_{x\in[N]}$. Then there exists an $N$-partite system $(\M_x)_{x\in[N]}$ with an embezzling state $\omega$ such that $\M^{(n)}_x \subset \M$ and $\omega \restriction \M^{(n)}_{[N]} = \omega^{(n)}$. 
\end{introtheorem}
In fact the resulting multipartite system $(\M_x)_{x\in[N]}$ with embezzling state $\omega$ is unique up to isomorphism.
Theorems \ref{introthm:state-to-family} and \ref{introthm:family-to-state} clarify the relation between embezzling states and embezzling families and provide an answer to question 2.
The assumption of consistency is key to this result: In \cref{sec:vdh} we show how to interpret the vDH family as an embezzling family on spin chains of increasing lengths. 
However, it is not consistent. Indeed, we prove:

\begin{introproposition}[informal, see~\cref{prop:vDHmbz}]\label{introthm:vdh}
    The van Dam-Hayden embezzling family on a spin-1/2 chain converges to a pure state $\Omega$ on a bipartite system $(\M_1,\M_2)$, which supports no embezzling states.
\end{introproposition}

As mentioned in the introduction, the LTW embezzling family can be generalized to the multipartite setting. In fact, it is a consistent embezzling family. We use these properties in combination with \cref{introthm:family-to-state} to show:
\begin{introtheorem}[informal,  see~\cref{thm:ltw_type}]\label{introthm:ltw}
For any $N\in\NN$ with $N\geq 2$, the LTW construction induces a multipartite system $(\M_x)_{x\in[N]}$ on a Hilbert space $\H$ with the following properties:
\begin{enumerate}
\item The multipartite system is irreducible, $\bigvee_x \M_x = \B(\H)$, and fulfills \emph{Haag duality}:
\begin{align}
    \bigg(\bigvee_{x\in I} \M_x\bigg)' = \bigvee_{y\notin I }\M_y,\qquad I \subset[N].
\end{align}
\item All density operators $\rho$ on $\H$ are multipartite embezzling states.

\item All factors $\M_x$ have type $\III_1$. 
\end{enumerate}
\end{introtheorem}
The theorem answers the third question affirmatively: Not only do there exist multipartite embezzling states, but even multipartite systems where \emph{all states} are embezzling.
Such a system will be called a \emph{universal multipartite embezzler}.

\cref{introthm:vdh} and \cref{introthm:ltw} together show that the van Dam-Hayden family and the Leung-Toner-Watrous family are indeed fundamentally different: One converges to a universal embezzler while the other does not even converge to an embezzling state.

The construction of the LTW embezzling family is, in a sense, highly inefficient as it involves tensor products over $\eps$-covers of unit spheres in arbitrarily high dimensional Hilbert spaces. In \cref{sec:vdh}, we show that, in the bipartite setting, one can construct a simpler consistent bipartite embezzling family that only involves product states and maximally entangled states of arbitrary dimensions in its construction. It relies on the following result:

\begin{introlemma}[informal, cp.~\cref{lem:mbz_traces}] A bipartite embezzling family that can embezzle maximally entangled states of arbitrary dimension can embezzle arbitrary pure bipartite states of arbitrary dimension.
\end{introlemma}

We emphasize that since we consider embezzlement relative to local operations and not relative to local operations and classical communication (LOCC), the above Lemma does not follow from the fact that any pure finite-dimensional entangled state can be obtained from a sufficiently high-dimensional maximally entangled state by LOCC.  
In particular, it is, in fact, not sufficient if the embezzling family allows to embezzle Bell states (see \cref{rem:traces}). 

For bipartite entanglement, \cite{long_paper} establishes precise quantitative criteria on an operator algebraic level for the existence of embezzling states on bipartite quantum systems. In particular, type $\III_{1}$ von Neumann algebras are singled out as the local von Neumann algebras of universal embezzlers. In contrast, in bipartite systems composed of commuting factors in Haag duality, all pure states are embezzling states \emph{relative to LOCC} if and only if the factors have type $\III$, irrespective of the subtype \cite{van_luijk_pure_2024}.
\begin{introquestion}
    What are the necessary and sufficient conditions for a multipartite system $(\M_1,\ldots,\M_N)$ to host multipartite embezzling states? Can we obtain a quantitative theory of multipartite embezzlement?
\end{introquestion}
As we will see below, in contrast to the bipartite setting, the existence of multipartite embezzling states cannot be decided solely by looking at the type of the local von Neumann algebras, even if each of them is a factor in standard representation. 
Rather, it depends on how each factor $\M_x$ is embedded in $\B(\H)$ relative to the remaining factors $(\M_y)_{y\neq x}$.
Hence, multipartite embezzlement is  connected to the theory of subfactor inclusions \cite{kawahigashi_subfactor_2005,longo_nets_1995,Jones_Sunder_1997} and we expect that its study will lead to fruitful results at the intersection of multipartite entanglement theory and the theory of operator algebras.

It follows from results in \cite{long_paper} that the vacuum sector of relativistic quantum field theories are bipartite embezzling states if we partition spacetime into two causally closed spacetime regions. 
In fact, they are universal bipartite embezzlers.
Generalizing a result of Matsui \cite{matsui_split_2001,keyl2006}, we have shown in \cite{van_luijk_critical_2024} that a large class of one-dimensional critical many-body systems provides bipartite, universal embezzlers in their ground state sectors. 
If we work in one spatial dimension and assume that the subsystems of the parties are simply connected, neither the vacua of relativistic QFTs on 1+1 dimensional Minkowski space fulfilling the \emph{split property} \cite{buchholz_product_1974,buchholz_causal_1986,buchholz_noethers_1986,werner_local_1987} nor ground states of spin chains  provide $N$-partite embezzling states with $N>2$.
This follows because we show that a multipartite embezzling state induces a bipartite embezzling state on any pair of subsystems. 
However, in the case of QFT, the split property implies that two properly spacelike separated regions cannot host bipartite embezzling states (see \cref{cor:nogo}).
Similarly, in a spin chain, $N-2$ parties only have access to a finite-dimensional system, which again precludes embezzlement because bipartite systems cannot host embezzling states if at least one of the subsystems is finite-dimensional  \cite{long_paper}.
In view of the existence of multipartite embezzling states, this leads to the following open questions:

\begin{introquestion} 
    Are the vacua of relativistic quantum field theories in $D>1$ spatial dimensions multipartite embezzling states? Are there examples of local Hamiltonians on a $D$-dimensional lattice of spins whose ground state are a multipartite embezzling state?
\end{introquestion}

\section{Multipartite embezzling families and states}

\subsection{Multipartite systems with infinitely many degrees of freedom.}

\begin{definition}[Multipartite systems]\label{def:N partite}
    Let $N$ be an integer.
    An \emph{$N$-partite system} is a collection $(\M_{1},\ldots,\M_{N})$ of $N$ pairwise commuting von Neumann subalgebras $\M_{x}\subset \M_{[N]}$ of a von Neumann algebra $\M_{[N]}$ such that $\M_{[N]}=\vee_{x=1}^N \M_x$.
    
    An $N$-partite system is \emph{factorial} if each $\M_x$, and hence $\M_{[N]}$, is a factor.
    If $\M_{[N]}$, and hence each subalgebra $\M_x$, acts on a Hilbert space $\H$, we say that $(\M_x)_{x=1}^N$ is an \emph{$N$-partite system on $\H$}.
    An $N$-partite system on $\H$ is \emph{irreducible} if $\M_{[N]}=\B(\H)$.
    For $I \subset [N]$, we write $\M_I = \bigvee_{x\in I} \M_x$.
    An $N$-partite system on $\H$ satisfies \emph{Haag duality} if
    \begin{equation}\label{eq:rHD}
        \M_I' = \M_{I^c}
    \end{equation}
    for all $I\subset [N]$.
    An $N$-partite system $(\M_x)_{x\in[\N]}$ is a \emph{subsystem} of an $N$-partite system $(\N_x)_{x\in[N]}$ if $\M_x\subset \N_x$ for all $x\in[N]$.
\end{definition}

\begin{definition}[Multipartite states]
    A \emph{state} on a multipartite system $(\M_1,\ldots,\M_N)$ is a normal state $\omega$ on $\M_{[N]}$.
    For an irreducible multipartite system on a Hilbert space $\H$, we identify density operators $\rho$ on $\H$ with the induced states $\omega = \tr \rho(\placeholder)$ and pure states with unit vectors $\Omega\in\H$ (up to phase) via $\omega = \bra \Omega\placeholder \ket\Omega$.
    If $\omega$ is a state on a multipartite system, and $I\subset [N]$ is a collection of parties, we write $\omega_I$ for the reduced state $\omega_I :=\omega \upharpoonright \M_I$ on $\M_I\subset\M_{[N]}$.
\end{definition}

For each $x$, the von Neumann algebra $\M_x$ contains the observables that are accessible to the party $x$, and $\M_{[N]}$ (the `$N$-partite algebra') is the observable algebra of the joint system consisting of all $N$ parties.
A multipartite system is an $N$-partite system for some integer $N\ge2$.
A factorial finite-dimensional $N$-partite system with local Hilbert space dimensions $d_x$ is given by $(M_{d_x}(\CC))_{x\in[N]}$ with 
\begin{align}
    \M_{[N]} = \bigotimes_{x=1}^N M_{d_x}(\CC). 
\end{align}
We will also simply refer to it as the \emph{$N$-partite system on $\otimes_{x=1}^N\CC^{d_x}$}.

In the case of a multipartite system on a Hilbert space $\H$, the factor $\M_{[N]}$ can be recovered from the subfactors $\M_x$ in their action on $\H$ via $\M_{[N]}= \bigvee_{x=1}^N \M_x$.

If $(\M_1,\ldots,\M_N)$ is an $N$-partite system, the joint observable algebra describing a collection $I\subset [N]$ of parties, is given by the von Neumann algebra
\begin{equation}
    \M_I := \bigvee_{x\in I}\M_x.
\end{equation}
This is consistent with $\M_{[N]}$ being generated by the subalgebras $\M_x$, $x\in [N]$.
In the case that $I=\emptyset$, we set $\M_\emptyset=\CC1$.

\begin{lemma}
    Let $(\M_x)_{x\in[N]}$ be an $N$-partite system on $\H$.
    If $(\M_x)_{x\in[N]}$ is irreducible, then it is factorial.
    If it is factorial and satisfies Haag duality, then it is irreducible.
\end{lemma}
\begin{proof}
    Assume that $(\M_x)_{x\in[N]}$ is irreducible and let $a \in \M_y$ be a central element for some $y$.
    Since $\M_{x}$ commutes with $\M_{y}$ for all $y\ne x$, we get that $a$ is also central in $\vee_x \M_x = \M_{[N]} = \B(\H)$, which implies that $x$ is proportional to the identity.
    If Haag duality and factoriality hold, then $\B(\H) = \M_x \vee \M_x' = \M_x \vee \M_{\{x\}^c} = \vee_y \M_y = \M_{[N]}$.
\end{proof}

\begin{definition}\label{def:split}
    A bipartite system $(\M_1,\M_2)$ is \emph{split} (or has the split property) if the bipartite algebra is isomorphic to the tensor product: $\M_1\vee \M_2 \cong \M_1\ox \M_2$ (with the isomorphism that sends products $ab$ to the elementary tensors $a\ox b$ for all $(a,b)\in\M_1\times\M_2$).
\end{definition}

If a bipartite system $(\M_1,\M_2)$ on a Hilbert space $\H$ is split then we can find  a tensor product decomposition $\H = \H_1\ox \H_2$ and von Neumann algebras $\N_j\cong \M_j$ on $\H_j$, $j=1,2$, such that\footnote{More precisely, there is a unitary $u:\H \to \H_1\ox\H_2$ such that $u^*\M_1u = \N_1\ox1$ and $u^*\M_2u=1\ox\N_2$.}
\begin{equation}\label{eq:vna_split_tp}
    (\M_1,\M_2) = (\N_1\ox1,1\ox \N_2).
\end{equation}
If $(\M_1,\M_2)$ is a bipartite system, then a state $\omega$ on the bipartite algebra $\M_1\vee\M_2$ is called a \emph{product state} if $\omega(ab)=\omega(a)\omega(b)$ for all $(a,b)\in\M_1\times\M_2$.

\begin{lemma}[{\cite[Remark]{dantoni_interpolation_1983}}]\label{lem:product-state}
    A bipartite system $(\M_1,\M_2)$ is split if and only if there exists a normal product state on $\M_1\vee \M_2$ with central support $1$ if and only if there is a faithful normal product state. 
\end{lemma}

\begin{lemma}\label{lem:split}
    Let $(\M_1,\M_2,\M_3)$ be a tripartite system on $\H$ fulfilling Haag duality.
    If the bipartite system $(\M_1,\M_3)$ is split, then there exists a tensor product decomposition $\H = \H_L\ox \H_R$ and von Neumann algebras $\N_1\cong \M_1$ on $\H_L$, $\N_3\cong \M_3$ on $\H_R$ such that 
    \begin{equation}
        (\M_1,\M_2,\M_3) = (\N_1\ox 1,\,\N_1'\ox \N_3',\, 1\ox \N_3).
    \end{equation}
\end{lemma}
\begin{proof}
     We apply \eqref{eq:vna_split_tp}, and use Haag duality to find $\M_2 = (\M_1\vee \M_3)' = (\N_1\ox \N_3)' = \N_1' \ox \N_3'$.
\end{proof}

In the bipartite case (i.e., $N=2$), the above definition of a bipartite system is exactly the von Neumann algebraic version of the $C^*$-algebraic definition of bipartite systems in \cite{van_luijk_schmidt_2023} (see also \cite{summers1990on_independence}).
If $\A_1,\A_2\subset \A$ is a $C^*$-algebraic bipartite system and if $\omega$ is a bipartite (pure) state in the sense of \cite{van_luijk_schmidt_2023}, then the GNS representation yields an (irreducible) bipartite system $(\M_1,\M_2)$ on the GNS Hilbert space $\H$ \cite[Prop.~16]{van_luijk_schmidt_2023}.

\begin{remark}[Bipartite system and inclusions of von Neumann algebras]
    Bipartite systems $(\M_1,\allowbreak \M_2)$ on a Hilbert space $\H$ are in one-to-one correspondence with inclusions, $\N\subset \M$, of von Neumann algebras via $\M_1 =\N$, $\M_2=\M'$.
    This correspondence is such that the bipartite system is irreducible/split if and only if the inclusion is irreducible/split.\footnote{An inclusion $\M\subset\N$ of von Neumann algebras (or $W^*$-inclusion) on $\H$ is irreducible if the relative commutant is trivial $\M'\cap \N=\CC1$, split if there exists a type I factor $\P$ such that $\M\subset\P\subset \N$ and trivial if $\M=\N$ \cite{doplicher1984split}.}
    The inclusion is trivial ($\N=\M$) if and only if Haag duality holds $\M_1=\M_2'$.
\end{remark}

We now consider tensor products of multipartite systems.
Let $(\M\up k_1,\ldots,\M\up k_N)$, $k\in K$, be a collection of $N$-partite systems.
In case of finitely many multipartite systems $\# K<\oo$, we can simply form the von Neumann tensor product $\M_x=\bigotimes_{k\in K}\M\up k_x$ of the von Neumann algebras corresponding to each of the parties to obtain an $N$-partite system $(\M_1,\ldots,\M_N)$ with $\M_{[N]} = \bigotimes_{k\in K} \M_{[N]}\up k$.
In the case of infinitely many $N$-partite systems, we have to be careful.

For each $k\in K$, let $\omega\up k$ be a normal state on $\M_{[N]}\up k$ with faithful GNS representation (i.e., the central support of $\omega$ is $1$).
Relative to such a state, we now form the incomplete infinite tensor product of the $N$-partite algebras $\M_{[N]}\up k$ to obtain a von Neumann algebra
\begin{equation}
    \M_{[N]} = \bigotimes_{k\in K} (\M_{[N]}\up{k};\omega\up k).
\end{equation}
We now define for each $x$ a subalgebra 
\begin{equation}
    \M_x = \bigvee_{k\in K} \M_x\up k \subset \M_{[N]}
\end{equation}
to obtain an $N$-partite system $(\M_1,\ldots,\M_N)$. We will denote this $N$-partite system as
\begin{equation}
    (\M_1,\ldots,\M_N) = \bigotimes_{k\in K} (\M_1,\ldots, \M_N ; \omega\up k).
\end{equation}

\begin{lemma}\label{lem:haag}
    Consider each of the $N$-partite systems $(\M_1\up k,\ldots,\M_2\up k)$ on the GNS Hilbert space $(\H\up k,\Omega\up k)$ of $\omega\up k$ and consider $(\M_1,\ldots,\M_N)$ in its action on $\H=\otimes_k (\H\up k,\Omega\up k)$.
    If Haag duality holds for each $k$, then it also holds for $(\M_1,\ldots,\M_N)$.
\end{lemma}

\begin{proof}
    Note that it suffices to show the bipartite case since, for each $I \subset [N]$, the bipartite systems $(\M_I\up k,\M_{I^c}\up k)$ satisfy Haag duality.
    The claim follows from \cite[Thm. XIV.1.9]{takesaki3}.
\end{proof}

The most important application of the Lemma arises if each $(\M\up{k}_x)_{x\in[N]}$ is a factorial, finite-dimensional quantum system, which ensures Haag duality.

\subsection{Multipartite embezzling states}

In the following we define $N$-partite embezzling states and  families. We then discuss the relations between the two concepts.

\begin{definition}[Multipartite embezzling states]\label{def:mbz-state}
    Let $(\M_x)_{x\in[N]}$ be an $N$-partite system.
    A multipartite state $\omega$ is \emph{($N$-partite)} embezzling if for all $d\in\NN$, every two pure states $\psi$ and $\phi$ on the $N$-partite system on $(\CC^d)^{\otimes N}$ and every $\eps>0$ there exist unitaries $u_x\in \M_x \ox M_d(\CC)$ such that
\begin{equation}\label{eq:mixed state mbz}
    \norm{u (\omega\ox \phi)u^* - \omega \ox \psi} < \eps,\quad u=\prod_x u_x.
\end{equation}
For an $N$-partite system $(\M_1,\ldots,\M_N)$ on $\H$ we say that $\Omega\in\H$ is embezzling if $\bra\Omega\placeholder\ket\Omega$ is embezzling.
\end{definition}

In the case of bipartite systems, this definition is more general than the one given in \cite{long_paper}: First of all, it applies to a more general definition of bipartite system (in \cite{long_paper}, we only considered bipartite systems in Haag duality on a Hilbert space) and, second, it allows for the bipartite state to a be general state on the bipartite algebra instead of restricting to vector states.
The marginal states $\omega_x$, $x\in [N]$, of an $N$-partite emezzling state $\omega$ satisfy a monopartite version of the embezzling property, which can be viewed as a form of approximate catalysis:

\begin{definition}\label{def:mono}
    Let $\M$ be von Neumann algebra. A normal state $\omega$ on $\M$ is \emph{monopartite embezzling} if for every $d\in\NN$, every $\eps>0$ and every two states $\psi$ and $\phi$ on $M_d(\CC)$ there exists a unitary $u\in \M\ox M_d(\CC)$ such that
    \begin{align}
        \norm{u(\omega\ox\phi)u^* - \omega\ox\psi} <\eps.
    \end{align}
\end{definition}

Note that this definition of monopartite embezzling states on a von Neumann algebra $\M$ is \emph{not} a special case of the definition of an $N$-partite embezzling state with $N=1$ since we allow for mixed states $\psi,\phi$ in \cref{def:mono} but only consider pure states $\psi,\phi$ in \cref{def:mbz-state}. 
Instead, we will see that each of the marginals of a multipartite embezzling state is a monopartite embezzling state.
Assuming Haag duality, monopartite and bipartite embezzling are equivalent:

\begin{proposition}[{\cite[Thm.~11]{long_paper}}]
    Let $(\M_1,\M_2)$ be a bipartite system on $\H$ satisfying Haag duality, let $\Omega\in \H$ be a state vector and denote by $\omega_x$ the induced marginal state on $\M_x$, $x=1,2$.
    Then $\Omega$ is (bipartite) embezzling if and only if $\omega_1$ monopartite embezzling if and only if $\omega_2$ is monopartite embezzling.
\end{proposition}

A direct consequence of the definitions is the following:

\begin{lemma}\label{lem:subsystems-mbz}
    Let $(\M_x)_{x\in[N]}$ be an $N$-partite system and $\omega$ an embezzling state. 
    Let $I_m\subset [N]$, with $m=1,\ldots,M$ be a collection of disjoint subsets and $I=\cup_m I_m$. Then $\omega_I$ is an embezzling state on the $M$-partite system $(\M_{I_1},\ldots,\M_{I_M})$.
\end{lemma}
\begin{proof}
   Let $\psi_I,\phi_I$ be pure states on $M_d(\CC)^{\otimes M}$ and let $\psi = \psi_I \ox (\bra1\placeholder\ket1)^{\ox N-M}$, $\phi = \phi_I \ox (\bra1\placeholder\ket1)^{\ox N-M}$.  
   Let $u_x \in \M_x$ be unitaries that embezzle the pure state $\psi$ from the pure state $\phi$ with precision $\eps>0$. Set $u_{I_m} = \prod_{x\in I_m} u_x$ and $u_I = \prod_{x\in I}u_{x} = \prod_{m=1}^M u_{I_m}$. Then 
    \begin{align}
         \norm{u_I(\omega_I \ox \phi_I)u_I^* - \omega_I\ox \psi_I} \leq  \norm{u(\omega \ox \phi)u^* - \omega\ox \psi} <\eps,
    \end{align}
    where we used monotonicity of the norm.
\end{proof}

\begin{lemma}\label{lem:pure-to-mixed}
Let $(\M_x)_{x\in [N]}$ be an $N$-partite system, $\omega$ an embezzling state and $I\subset[N]$ a collection of $M\leq \lfloor N/2\rfloor$ subsystems.
Then the state $\omega_I$ on the $M$-partite system $(\M_x)_{x\in I}$ can embezzle arbitrary $M$-partite mixed states: For all $d\in\NN$, all states $\psi,\phi$ on $M_d(\CC)^{\ox M}$, and every $\eps>0$ there exist unitaries $u_x \in \M_{x}$, $x\in I$, such that
\begin{equation}
    \norm{u_I (\omega_I\ox\phi)u_I^*-\omega\ox\psi}<\eps, \qquad u_I=\prod_{x\in I}u_x.
\end{equation}
In particular, each of the marginals $\omega_x$, $x\in[N]$, is a monopartite embezzling state.
\end{lemma}
\begin{proof}
    Let $\psi$ be a state on $M_d(\CC)^{\otimes M}$. Since $M\leq \lfloor N/2\rfloor$, there exists a pure state $\varphi$ on $M_d(\CC)^{\otimes N}$, so that $\psi = \varphi_I$. 
    Now apply the same reasoning as in \cref{lem:subsystems-mbz}.
\end{proof}

\begin{corollary}\label{cor:nogo} Consider an $N$-partite system $(\M_x)_{x\in[N]}$ with $N\geq 3$. If any of the bipartite subsystems $(\M_x,\M_y)$ with $x\ne y$ is split, then $(\M_x)_{x\in[N]}$ cannot host $N$-partite embezzling states. 
\end{corollary}
\begin{proof}
   By \cref{thm:no-go}, there exist no embezzling states on $(\M_x,\M_y)$. Hence, by \cref{lem:subsystems-mbz} the system $(\M_x)_{x\in[N]}$ cannot host $N$-partite embezzling states.
\end{proof}

\subsection{Multipartite embezzling families}

\begin{definition}\label{def:increasing-net-Npartite}
    An \emph{increasing net of $N$-partite systems} $(\M^{(n)}_1,\ldots,\M^{(n)}_N)_n$ is a collection of $N$ increasing nets of von Neumann algebras $(\M_x^{(n)})_n$ such that $(\M^{(n)}_x)_{x\in[N]}$ is an $N$-partite system for each $n$. 
    An $N$-partite system $(\M_x)_{x\in[N]}$ is \emph{generated} by an increasing net of subsystems $(\M^{(n)}_1,\ldots,\M^{(n)}_N)_n$ if $\M_x = \vee_n \M^{(n)}_x$ for all $x\in[N]$.
\end{definition}

\begin{definition}[Multipartite embezzling family]\label{def:multipartite-mbz-family}
  A net of normal states $(\omega^{(n)})_n$ on an increasing net of $N$-partite systems $(\M_1^{(n)},\ldots,\M_N^{(n)})_n$ is an $N$-\emph{partite embezzling family} if for all $d\in\NN$ and $\eps>0$ there exists $n = n(d,\eps)$ such that for all $m\geq n$ and all pure states $\psi$ and $\phi$ on the $N$-partite system on $(\CC^d)^{\otimes N}$ there exist unitaries $u^{(m)}_x \in \M^{(m)}_x\otimes M_{d}(\CC)$ such that
  \begin{align}
      \norm{u^{(m)}(\omega^{(m)}\ox\phi)(u^{(m)})^* - \omega^{(m)} \ox \psi} < \eps,\quad u^{(m)} = \prod_x u_x^{(m)}.
  \end{align}
\end{definition}

If we have an increasing net of bipartite systems, local unitaries $u^{(m)}_x \in \M^{(m)}_x$ are also elements of $\M^{(n)}_x$ for $n\geq m$ and the states $\omega^{(n)}$ on $\M^{(n)}_{[N]} = \vee_x \M^{(n)}_x$ restrict to $\M^{(m)}_{[N]}$ for $m\leq n$.  This allows us to define quasi-consistent  and consistent embezzling families:

\begin{definition}[Consistency]
\label{def:consistent}
    An $N$-partite embezzling family $(\omega^{(n)})_n$ on an increasing net of $N$-partite systems $(\M_1^{(n)},\ldots,\M_N^{(n)})_n$ is \emph{quasi-consistent} if for all $d\in\NN$ and $\eps>0$ there exists $n = n(d,\eps)$ such that for all pure states $\psi$ and $\phi$ on the $N$-partite system on $(\CC^d)^{\otimes N}$ there exist unitaries $u^{(n)}_x \in \M^{(n)}_x\otimes M_{d}(\CC)$  such that for all $m\geq n$
  \begin{align}
      \norm{u^{(n)}(\omega^{(m)}\ox\phi)(u^{(n)})^* - \omega^{(m)} \ox \psi} < \eps,\quad u^{(n)} = \prod_x u_x^{(n)}.
  \end{align}
  If in addition we have $\omega^{(n)} = \omega^{(m)} \restriction \M^{(n)}_{[N]}$ for all $n\leq m$, we call $(\omega^{(n)})_n$ \emph{consistent}.
\end{definition}
Now consider a state $\omega$ on a multipartite system $(\M_x)_{x\in[N]}$ and assume that each $\M_x$ contains a weakly dense increasing net of von Neumann algebras:\footnote{We can assume, without loss of generality, that the index set $I$ of the net is the same for each party: Otherwise, construct the product $I=I_1\times \ldots\times I_N$ of the local sets $I_x$ with the product order and set $\M^{(n)}_x = \M^{(n_x)}_x$ for $n=(n_1,\ldots,n_N)\in I$. }
\begin{align}
    \M_x = \bigvee_n \M_x^{(n)},
\end{align}
meaning that the multipartite system $(\M_x)_{x\in[N]}$ is generated by  the increasing net of subsystems $(\M\up{n}_1,\ldots,\M\up{n}_N)_n$.
It follows that, if $\omega$ is embezzling, then the obtained sequence is an embezzling family (see \cref{thm:state-to-family} below).

A prototypical example of this setting is an infinite lattice system, partitioned into $N$ infinite cones in a "piece of cake"-like partitioning (see \cref{fig:illustration}). Then, each cone contains a weakly dense, increasing sequence of finite type I factors, each describing finitely many of the spins of the cone. More generally, whenever each $\M_x$ is hyperfinite (or approximately finite-dimensional), this is true.
By restricting $\omega$ to the multipartite systems
\begin{align}
    (\M^{(n)}_x)_{x\in[N]},\quad \M^{(n)}_{[N]} = \vee_x \M^{(n)}_x,
\end{align}
we obtain a consistent sequence of states $(\omega^{(n)})_n$ on an increasing sequence of $N$-partite systems.

\begin{theorem}\label{thm:state-to-family}
    Let $(\M_x)_{x\in[N]}$ be an $N$-partite system generated by an increasing net of subsystems $(\M\up{n}_x,\ldots,\M\up{n}_N)_n$.
    If $\omega$ is an $N$-partite embezzling state, then its restrictions $(\omega^{(n)})_n$ to $(\M^{(n)}_x)_{x\in[N]}$ form a consistent $N$-partite embezzling family.
\end{theorem}
The theorem generalizes \cite[Thm. 35]{van_luijk_critical_2024} to the multipartite setting. 
Since its proof is essentially identical (up to replacing sequences with nets and taking products of unitaries), we omit it here.
It relies on the observation that each party may approximate their local embezzling unitaries $u_x$ arbitrarily well by unitary elements of the increasing net of factors $\cup_n \M^{(n)}_x$:\footnote{In \cite{van_luijk_critical_2024} the Lemma is only stated for increasing sequences, but it is also valid for increasing nets.}
\begin{lemma}[\cite{van_luijk_critical_2024}, Lem.~36]
\label{lem:unitary_closure}
Consider a weakly dense, increasing net of von Neumann algebras $\cup_n \M^{(n)}\subset \M$ and denote $\M_{0}=\bigcup_{n}\M_{n}$. Then
\begin{align}
\label{eq:unitary_closure}
\U(\M) & = \overline{\U(\M_{0})}^{s^*}
\end{align}
where the closure on the right is taken in the strong* topology.
\end{lemma}
We now consider the converse question: Does a consistent or quasi-consistent embezzling family induce an embezzling state?
\begin{theorem}\label{thm:consistent-to-state}
    Let $(\M^{(n)}_1,\ldots,\M^{(n)}_N)_{n}$ be an increasing net of $N$-partite systems and let $(\omega^{(n)})_{n}$ be a consistent $N$-partite embezzling family.
    Then there exists an $N$-partite system $(\M_x)_{x\in[N]}$ with an embezzling state $\omega$ and a family of normal *-morphisms $j^{(n)}:\M^{(n)}_{[N]}\to \M_{[N]}$ such that:
    \begin{enumerate}
    \item $j^{(n)}\restriction \M^{(m)}_{[N]} = j^{(m)}$ for $m\leq n$,
    \item  $j^{(n)}(\M^{(n)}_x) \subset \M_x$, 
    \item $\M_{[N]} = \bigvee_n j^{(n)}(\M^{(n)}_{[N]})$ and  $\omega$ has central support $1$,
    \item $\omega^{(n)} = \omega\circ j^{(n)}$.
    \end{enumerate}
    With these properties, the triple $(\M_{[N]},(j^{(n)})_n,\omega)$ is unique up to isomorphism.
\end{theorem}
An important special case arises when the increasing net of $N$-partite systems is factorial: since normal *-morphisms from a  factor into another von Neumann algebra are injective, we can then view each $\M^{(n)}_{[N]}$ as embedded in $\M_{[N]}$:
\begin{corollary}\label{cor:consistent-to-state-factorial}
    Let $(\M^{(n)}_1,\ldots,\M^{(n)}_N)_{n}$ be an increasing net of factorial $N$-partite systems and let $(\omega^{(n)})_{n}$ be a consistent $N$-partite embezzling family. There exists an $N$-partite system $(\M_x)_{x\in[N]}$ generated by  $(\M^{(n)}_1,\ldots,\M^{(n)}_N)_{n}$, together with an embezzling state $\omega$ such that $\omega^{(n)} =\omega\restriction \M^{(n)}_{[N]}$. With these properties the pair of $N$-partite system and embezzling state is unique up to isomorphism. 
\end{corollary}

To prove \cref{thm:state-to-family}, we need an auxiliary result. This is a direct consequence of \cite[Thm.~IV.9.5]{glimm1985qft_expositions}. We include proof for the sake of completeness.

\begin{lemma}\label{lem:local_normal}
Let $(\M^{(n)})_{n}$ be an increasing net of von Neumann algebra. Let $\omega_0$ be a state on $\M_{0} = \bigcup_{n}\M\up{n}$ such that $\omega^{(n)}=\omega_0\restriction\M\up{n}$ is normal for all $n$. Then, $\pi\restriction\M\up{n}$ is normal for all $n$, where $(\H,\pi,\Omega)$ denotes the GNS representation of $\omega_0$.
\end{lemma}
\begin{proof}
Consider a sequence of vectors $(\Phi_{j})_{j}$ given by $\Phi_{j} = \pi(a_{j})\Omega$ for $a_{j}\in\M_0$. 
By assumption, the positive functionals $\varphi_{j} = a_j\omega a_j^* = \bra\Omega\pi(a_{j}^{*})\placeholder\pi(a_{j})\ket{\Omega} = \bra{\Phi_{j}}(\placeholder)\ket{\Phi_{j}}$ are normal on any $\M^{(n)}$.
This allows us to conclude that any functional of the form $\varphi = \bra{\Phi}\placeholder\ket{\Phi}$ with $\Phi\in\H$ is normal for any $\M^{(n)}$ as follows: By construction $\Omega$ is cyclic for $\M_0$, i.e., we may assume that $\Phi$ is approximated by vectors of the form $\Phi_{j} = \pi(a_{j})\Omega$, which entails that
\begin{align*}
   \norm{\varphi_{j} - \varphi} & \leq (\norm{\Phi_{j}}+\norm{\Phi})\norm{\Phi_{j}-\Phi} \to 0, \qquad j \to \infty
\end{align*}
and, therefore, $\varphi$ is normal as norm-limits of normal states are normal \cite[Lem.~III.3.14]{takesaki1}. This implies that all maps of the form
\begin{align*}
   a & \mapsto \sum_{k=1}^{\infty}\bra{\Phi_{k}}\pi(a)\ket{\Phi_{k}}, & \Phi_{k}\in\H,\ \sum_{k=1}^{\infty}\norm{\Phi_{k}}^{2} & < \infty,
\end{align*}
is normal on any $\M^{(n)}$, i.e., $\pi\restriction\M^{(n)}$ is normal.
\end{proof}

\begin{proof}[Proof of \cref{thm:consistent-to-state}]
    Existence: We consider the normed unital *-algebra $\M_{[N],0} = \bigcup_{n}\M^{(n)}_{[N]}$ together with the state $\varphi$ given by $\varphi|_{\M^{(n)}_{[N]}} = \omega^{(n)}$. 
    Taking the GNS representation $(\H,\pi,\Omega)$ associated with $(\M_{[N],0},\varphi)$, we obtain a von Neumann algebra as the double commutant $\M_{[N]} = \pi(\M_{[N],0})''$.
    The *-morphisms $j^{(n)} = \pi|_{\M^{(n)}_{[N]}}$ are normal by \cref{lem:local_normal}. By construction 
    \begin{align}
    \M_{[N]}=\bigvee_n j^{(n)}(\M^{(n)}_{[N]}) = \bigvee_{n,x} j^{(n)}(\M^{(n)}_x) = \bigvee_x \M_x,\quad \M_x = \bigvee_n j^{(n)}(\M^{(n)}_x).
    \end{align}
    We can extend $\varphi$ to a define a normal state $\omega$ on $\M_{[N]}$ via
    \begin{align*}
        \omega(a) & = \bra{\Omega}\pi(a)\ket{\Omega}, &  a & \in\M_{[N]}.
    \end{align*}
    We show that $\omega$ is embezzling:
    By assumption, given $d\in\NN, \eps>0$, we find $n = n(d,\eps)$ such that for all pure states $\psi$ and $\phi$ on the $N$-partite systme on $(\CC^d)^{\ox N}$ there exist unitaries $u^{(n)}_x\in M_{d}(\M^{(n)}_x)$ such that
    \begin{align*}
        \norm{u^{(n)}(\omega^{(m)}\ox\phi)(u^{(n)})^{*}-\omega^{(m)}\ox\psi} & < \eps,\quad   m\geq n,\quad u^{(n)}=\prod_x u^{(n)}_x. 
    \end{align*}
    Now let $a\in \pi(\M_{[N],0})\ox M_d(\CC)^{\ox N}$ be a self-adjoint contraction and let $m\geq n$ be sufficiently large such that there is a self-adjoint contraction $a_0\in \M^{(m)}_{[N]}\ox M_d(\CC)^{\ox N}$ with $a=\pi(a_0)$ and $\norm{a}=\norm{a_0}\leq 1$. The existence of $a_0$ follows immediately from \cite[Prop.~II.5.1.5]{blackadar_operator_2006}.
    It follows that
    \begin{equation}
        \big|\big[ u^{(n)}(\omega \ox \phi)(u^{(n)})^* -\omega\ox\psi\big](a)\big|
        = \big|\big[u^{(n)}(\omega^{(m)}\ox \phi)(u^{(n)})^*-\omega^{(m)} \ox \psi \big](a_0)\big|< \eps.\nonumber
    \end{equation}
    Since $\pi(\M_{[N],0}))\ox M_d(\CC)^{\ox N}$ is ultraweakly dense, and since the states involved are normal, this implies $\norm{u^{(n)}(\omega \ox\phi)(u^{(n)})^*-\omega\ox\psi}<\eps$.
    Therefore, $\omega$ is embezzling.
    Since $\Omega$ is a cyclic vector for $\M_{[N]}$, the central support of $\omega$ in $\M_{[N]}$ is $1$.    
    
    Uniqueness:
    Let $\M_{[N]}, j^{(n)}, \omega$ be as specified.
    Since $\omega$ has central support $1$, its GNS representation is faithful.
    We may hence assume $\M_{[N]}\subset \B(\H)$ to be given in the GNS representation and let $\Omega$ denote the GNS vector. 
    We consider the unital *-algebra $\M_{[N],0}= \bigcup_n\M^{(n)}_{[N]}$ with the state $\varphi$ as above and define a *-morphism $\pi:\M_{[N],0}\to \M_{[N]}\subset\B(\H)$ via $\pi\upharpoonright \M^{(n)}_{[N]}=j^{(n)}$ and find $\varphi = \bra\Omega\pi(\placeholder)\ket\Omega$.
    The assumption $\M_{[N]} = \bigvee_n j^{(n)}(\M^{(n)}_{[N]}) \equiv \pi(\M_{[N],0})''$ guarantees that $[\pi(\M_{[N],0})\Omega] = [\M_{[N]}\Omega]=\H$.
    By uniqueness of the GNS representation, the triple $(\M_{[N]},(j^{(n)})_n,\omega)$ is isomorphic to the one constructed in the existence step.    
\end{proof}

If the index set is given by $\NN$ and each $\M^{(n)}_{[N]}$ is purely atomic\footnote{A von Neumann algebra $\M$ is called purely atomic if it is a product of type I factors \cite{blackadar_operator_2006}.}, we can, in fact, show that embezzling states are \emph{equivalent} to quasi-consistent embezzling families. 
We first need the definition of a weak*-cluster point:

\begin{definition}\label{def:weak_star_cluster}
    Let $(\M^{(n)})_{n}$ be an increasing sequence von Neumann algebras and let $(\omega^{(n)})_{n}$ be a corresponding sequence of normal states. 
    We call a state $\omega$ on the unital *-algebra $\M_{0} = \bigcup_{n}\M^{(n)}$ a sequential weak*-cluster point if there is a subset $K\subseteq\NN$ such that
    \begin{align}
        \lim_{k\in K}(\omega^{(k)}-\omega)(a) & = 0, &  &  a\in\M_{0}.
        \end{align}
    In this case, we also say that $\omega$ is the weak*-limit of the subsequence $(\omega^{(k)})_{k\in K}$.
\end{definition}

\begin{proposition}\label{prop:mbzfam_quasi}
    Let $(\M_x)_{x\in[N]}$ be an $N$-partite generated by an increasing sequence of $N$-partite systems $(\M\up{n}_1,\ldots,\M\up{n}_N)_n$ with each $\M\up{n}_x$ purely atomic.
    Set $\M_{[N],0} = \cup_n \M^{(n)}_{[N]}$.
    Given a normal state $\omega$ on $\M_{[N]}$ such that $\omega_{0} = \omega|_{\M_{[N],0}}$ is a sequential weak*-cluster point of a sequence of normal states $(\varphi^{(n)})_{n}$ on $(\M^{(n)}_{[N]})_{n}$, the following are equivalent:
    \begin{enumerate}
        \item \label{it:mbzfam_quasi_1} $\omega$ is a multipartite embezzling state.
        \item \label{it:mbzfam_quasi_2} There exists  a (non-decreasing) sequence $(k_{n})_{n}\subset\NN$ such that $\omega_{0}$ is the weak*-limit of $(\varphi^{(k_{n})})_{n}$ and $(\varphi^{(k_{n})}\restriction \M^{(n)}_{[N]})_{n}$ is an $N$-partite quasi-consistent embezzling family.
    \end{enumerate}
\end{proposition}
\begin{proof}
In the following we use that $\M\up{n}_{[N]}=\vee_{x\in[N]} \M^{(n)}_x$ is purely atomic if each of the $\M\up n_x$ is and write $\varphi\up{m}_n$ for $\varphi\up{m}\restriction \M\up{n}_{[N]}$. 
We first show that \cref{it:mbzfam_quasi_1} $\implies$ \cref{it:mbzfam_quasi_2}: As before we denote $\omega\up n = \omega\restriction \M\up n_{[N]}$. Since $\omega_{0}$ is a sequential weak*-cluster point, we know that there is an infinite subset $K\subseteq\NN$ such that $\lim_{k\in K}\norm{\varphi\up{k}_n-\omega\up{n}} = 0$ as weak convergence is equivalent to norm convergence for sequences of normal states on purely atomic von Neumann algebras \cite{DellAntonio1967normal_limits}. 
Now, we choose $k=k_{n}$ such that $\norm{\varphi\up{k_{n}}_n-\omega\up{n}}=\delta_{n}\to0$ for $n\to\infty$, and we set $\chi\up{n} = \varphi\up{k_{n}}_n$.
Since $(\omega\up{n})_{n}$ is consistent by \cref{thm:state-to-family}, given any $d\in\NN, \eps>0$, we find $n = n(d,\eps)$ such that for all pure states $\psi$ and $\phi$ on $M_{d}(\CC)^{\ox N}$ there exists unitaries $u\up{n}_x\in\U(M_{d}(\M\up{n}_x))$ such that for $m\geq n$
\begin{align*}
\norm{u\up{n}(\omega\up{m}\ox\phi)(u\up{n})^{*}-\omega\up{m}\ox\psi} & < \tfrac{\eps}{2},\quad u\up n = \prod_x u\up n_x.
\end{align*}
This implies that $(\chi_{n})_{n}$ is quasi-consistent because
\begin{align} \nonumber
\|\chi\up{m}\ox\psi-u\up{n}\big(\chi\up{m}\ox\phi\big)(u\up{n})^{*}\|  &\leq 2\norm{\chi\up{m}-\omega\up{m}}\nonumber +\norm{\omega\up{m}\ox\psi-u\up{n}\big(\omega\up{m}\ox\phi\big)(u\up{n})^{*}} \\ 
    & < 2\delta_{m} + \tfrac{\eps}{2}\label{eq:mbzfam_approx}  < \eps
\end{align}
for $m\geq n$ such that $\delta_{m}<\tfrac{\eps}{2}$.

To show \cref{it:mbzfam_quasi_2} $\implies$ \cref{it:mbzfam_quasi_1}, we set $\chi\up{n} = \varphi\up{k_{n}}|_{\M\up{n}_{[N]}}$ as before and note that 
\begin{equation}
\lim_{n}\norm{\chi\up{n}_m-\omega\up{m}}=0,\qquad \chi\up{n}_m = \chi\up{n}|_{\M\up{m}_{[N]}},
\end{equation}
for all $m$ by the same argument as in the first step. By \cref{thm:consistent-to-state}, it is sufficient to show that the sequence $(\omega\up{m})_{m}$ is quasi-consistent (as the state consistency holds by assumption). Since the sequence $(\chi\up{n})_{n}$ is quasi-consistent, this follows: Given $d\in\NN, \eps>0$, we find by assumption $n = n(d,\eps)$ such that for all pure states $\psi$ and $\phi$ on $M_{d}(\CC)^{\ox N}$ there exists unitaries $u\up{n}_x\in\U(M_{d}(\M\up{n}_x))$ such that for $m\geq n$
\begin{align*}
\norm{u\up{n}(\chi\up{m}\ox\phi)(u\up{n})^{*}-\chi\up{m}\ox\psi} & < \tfrac{\eps}{2},\quad u\up{n} = \prod_x u\up{n}.
\end{align*}
For $k\geq m\geq n$, this implies that
\begin{align*}
\|u\up{n}(\omega\up{m}&\ox\phi)(u\up{n})^{*}-\omega\up{m}\ox\psi\|\\& \leq 2 \norm{\chi\up{k}_m\!-\!\omega\up{m}} + \norm{u\up{n}(\chi\up{k}_m\!\ox\phi)(u\up{n})^{*}-\chi\up{k}_{m}}\ox\psi \\
& \leq 2 \norm{\chi\up{k}_{m}\!-\!\omega\up{m}} + \norm{u\up{n}(\chi\up{k}\ox\phi)(u_{n})^{*}-\chi\up{k}\ox\psi} \\
& < 2 \norm{\chi\up{k}_m\!-\!\omega\up{m}} + \tfrac{\eps}{2},
\end{align*}
Now, we can choose $k\geq m$ sufficiently large to achieve
\begin{align*}
\norm{u\up{n}(\omega\up{m}\ox\phi)(u\up{n})^{*}-\omega\up {m}\ox\psi} & < \eps,
\end{align*}
which completes the proof.
\end{proof}

As a final result of this subsection, we show that for increasing \emph{sequences of finite-dimensional $N$-partite systems}, i.e., where the local algebras $\M\up{n}_x$ are direct sums of full matrix algebras, we can always select subsequences from a quasi-consistent embezzling family that are close to consistent embezzling families in an appropriate sense, i.e., there is no need to assume the existence of a sequential weak*-cluster point as in \cref{prop:mbzfam_quasi}.\footnote{In the language of \cite{vanluijk2024jdyn}, we consider $j^{*}$-convergent sequences of states with respect to the inductive system $(\M\up{n}_{[N]})_{n}$.}

\begin{proposition}
\label{prop:lim_con_fam}
    Let $(\M\up{n}_1,\ldots\M\up{n}_N)_{n}$ be an increasing sequence of $N$-partite systems, where each $\M\up{n}_x$ is a finite-dimensional von Neumann algebra. Let $(\varphi\up{n})_{n}$ be a quasi-consistent  $N$-partite embezzling sequence of normal states. 
    Then, there exists a subset $K\subset \NN$ and a consistent $N$-partite embezzling sequence $(\omega\up n)_{n}$ on $(\M\up{n}_1,\ldots,\M\up{n}_N)_{n}$ that is norm close to suitable restrictions of the original quasi-consistent family:
    \begin{equation}
    \label{eq:norm_lim_con}
        \lim_{n}\norm{\varphi\up{k_{n}}|_{\M\up{n}_{[N]}}-\omega\up{n}} = 0,
    \end{equation}
    for a (non-decreasing) sequence $(k_{n})_{n}\subset K$.
\end{proposition} 

In fact, we show that,
\begin{align}
\label{eq:norm_lim}
    \lim_{k\in K} \norm{\varphi\up{k}|_{\M\up{n}_{[N]}}-\omega\up{n}} & = 0,
\end{align}
for all $n$ which is stronger than \cref{eq:norm_lim_con}.\footnote{\label{ft:norm_lim} To see this, note that \cref{eq:norm_lim} implies that we can find a sequence $(k_{n})_n$ such that $\norm{\omega\up{n}-\varphi\up{k_{n}}\restriction \M\up{n}_{[N]}}<\eps_{n}$ for a null sequence $(\eps_{n})_{n}$.}
By \cref{thm:consistent-to-state}, we know that the consistent $N$-partite embezzling family $(\omega\up{n})_{n}$ defines a $N$-partite embezzling state $\omega$, which is a sequential weak*-cluster point of the quasi-consistent $N$-partite embezzling family $(\varphi\up{n})_{n}$:
\begin{align}
    \lim_{k\in K}\varphi\up{k}(a_{n}) & = \omega\up{n}(a_{n}), & n\in\NN,\  a_{n}\in\M\up{n}_{[N]}.
\end{align}

\begin{figure}[t]
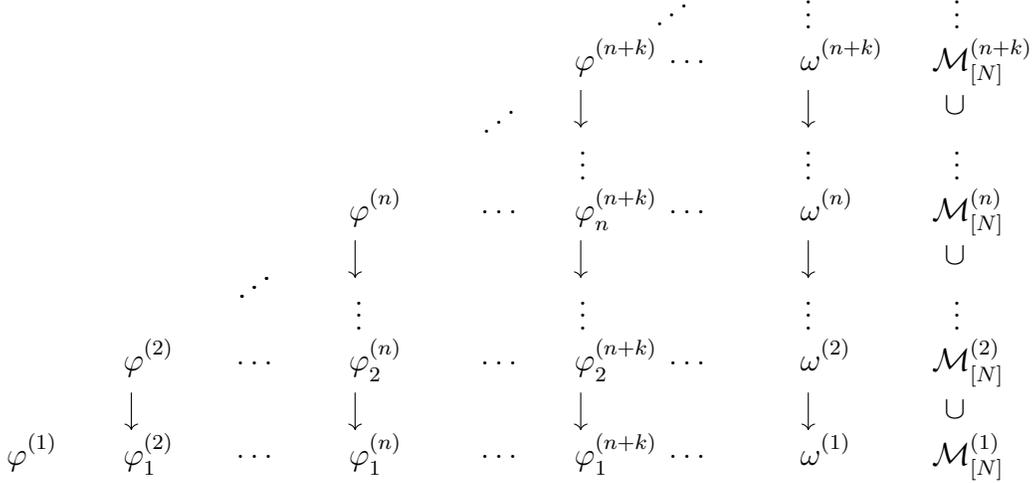

\begin{center}
\include{figure.tikz}
\caption{\small An illustration of the relations between the quasi-consistent embezzling family $(\varphi^{(n)})_{n}$, its restrictions $\varphi^{(n)}_{m} = \varphi\up{n}\restriction \M\up{m}_{[N]}$, and the consistent embezzling family $\omega\up{n}$ appearing in \cref{prop:lim_con_fam} (assuming that $K_n=\NN$).
The downward arrows symbolize the consistency relation $\varphi\up n_{m} \upharpoonright \M_{[N]}\up n = \varphi\up n_m$ for integers $n\ge m$ (with $\varphi\up n_n := \varphi\up n$).
}
\label{fig:Wtree}
\end{center}
\end{figure}
\begin{proof}
    For every $n$, we have the sequence of normal states $(\varphi\up{m}_n)_m := (\varphi\up{m}\restriction \M\up{n}_{[N]})_{m}$ which has a norm convergent subsequence $(\varphi\up{k_{n}}_n)_{k_{n}\in K_{n}}$, $K_{n}\subseteq\NN$, as each $\M\up{n}_{[N]}$ is finite dimensional (see \cref{fig:Wtree}). 
    Iterating over $n$, we can achieve that $K_{n}-1\supseteq K_{n+1}$\footnote{For a subset $K\subset \NN$, we denote by $K-1$ the subset where every element is reduced by $1$ (ignoring $0$ as a potential outcome).}, which allows us to form the diagonal subset $K$ consisting of the first element of each $K_{n}$. 
    It follows that the sequence of normal states $(\varphi\up{k}_n)_{k\in K}$ is norm convergent to a state $\omega\up{n}$ for all $n$:
    \begin{align}
        \lim_{k\in K} \norm{\varphi\up{k}_n-\omega\up{n}} & = 0, &  n\in\NN.
    \end{align}
    Since the states $\varphi\up{k}_n$ are normal for all $n,k$, we know that the limit states $\omega\up{n}$ are normal for all $n$ (the predual of a von Neumann algebra is weakly sequentially complete \cite[Cor.~III.5.2]{takesaki1}). 
    By construction, we have $\omega\up{n}|_{\M\up{m}_{[N]}} = \omega\up{m}$ for $n\geq m$ and, therefore, it remains to show that $(\omega\up{n})_{n}$ is a quasi-consistent embezzling family. 
    This follows by an analogous argument as in the proof of \cref{prop:mbzfam_quasi}: Given any $d\in\NN, \eps>0$, we find $n = n(d,\eps)$ such that for all pure states $\psi$ and $\phi$ on $M_{d}(\CC)^{\ox N}$ there exists unitaries $u\up{n}_x\in\U(M_{d}(\M\up{n}_x))$ such that for $k\geq n$
    \begin{align*}
    \norm{u\up{n}(\varphi\up{k}\ox\phi)(u\up{n})^{*}-\varphi\up{k}\ox\psi} & < \tfrac{\eps}{2},\quad \forall k\geq n,\quad u\up{n}=\prod_x u\up{n}_x.
    \end{align*}
    For $m\geq n$, this implies that
    \begin{align*}
    \| u\up{n}(\omega\up{m} &\ox\phi)(u\up{n})^{*}-\omega\up{m}\ox\psi \| \\ 
    &\leq 2 \norm{\varphi\up{k}_m\!-\!\omega\up{m}} + \norm{u\up{n}(\varphi\up{k}_m\ox\phi)(u\up{n})^{*}-\varphi\up{k}_m\ox\psi} 
    \\
    & \leq 2 \norm{\varphi\up{k}_m\!-\!\omega\up{m}} + \norm{u\up{n}(\varphi\up{k}\ox\phi)(u\up{n})^{*}-\varphi\up{k}\ox\psi} \\
    & < 2\norm{\varphi\up{k}_m\!-\!\omega\up{m}} + \tfrac{\eps}{2}.
    \end{align*}
    Because of \cref{eq:norm_lim}, we can choose $K\ni k\geq m$ sufficiently large to achieve
    \begin{align*}
    \norm{u\up{n}(\omega\up{m}\ox\phi)(u\up{n})^{*}-\omega\up{m}\ox\psi} & < \eps,
    \end{align*}
    which completes the proof.
\end{proof}

\begin{remark}
 An analog of \cref{lem:pure-to-mixed} obviously also holds for multipartite embezzling families instead of embezzling states using the same arguments.  Analogous results to \cref{thm:state-to-family}, \cref{prop:mbzfam_quasi}, and \cref{prop:lim_con_fam} hold in the monopartite setting as well as for the $N$-partite embezzlement of mixed states (as in \cref{lem:pure-to-mixed}) with essentially identical proofs (see \cref{def:mbzfam} for a definition of monopartite embezzling families).
\end{remark}

\subsection{The limit of the Leung-Toner-Watrous embezzling family}\label{sec:LTW}

So far, we discussed 
general relations between multipartite embezzling states and multipartite embezzling families. 
These results would be of little value unless multipartite states and families in fact exist. 
In \cite{long_paper,short_paper,van_luijk_critical_2024} we discussed how \emph{bipartite} embezzling states arise in quantum field theory and many-body physics. 
We will now show how the `universal construction' for embezzling families introduced by Leung, Toner, and Watrous (LTW) \ \cite{Leung2013coherent} gives rise to consistent, genuinely $N$-partite embezzling families for arbitrary $N\in \NN$. By \cref{thm:consistent-to-state}, we immediately deduce that $N$-partite embezzling states exist. In fact, we will show that the construction results in an irreducible, multipartite system on $\H$ that is a \emph{universal multipartite embezzler}: Every density matrix on $\H$ is an $N$-partite embezzling state.

In the following, we repeatedly consider multipartite systems where each party has access to several qudits. It is, therefore useful to introduce a short-hand for the Hilbert space describing $N$ parties, each consisting of $n$ qudits of Hilbert space dimension $d$:
\begin{equation}
    \K_{d,n,N}:= \bigotimes_{x=1}^N (\CC^d)^{\otimes n} \cong \bigotimes_{j=1}^{n} (\CC^d)^{\otimes N}\cong \CC^{d^{nN}}.
\end{equation}
To proceed with the construction, it will be useful to introduce an auxiliary object: First, given two unit vectors $\Psi,\Phi\in (\CC^d)^{\ox N}=\K_{d,1,N}$, we denote by $\lambda$ the phase such that $\ip{\lambda\Psi}{\Phi}\ge0$ (with $\lambda=1$ if the vectors are orthogonal) and define unit vectors
\begin{align}
    \Omega_n(\Psi,\Phi) & := \frac{1}{C_n}\sum_{j=0}^{n} \lambda^{n-j}\,
    {\Psi}^{\otimes j}\otimes {\Phi}^{\otimes n-j} \in \K_{d,n,N}
\end{align}
where $C_n$ is the normalizing constant ensuring that $\norm{\Omega_n(\Psi,\Phi)}=1$.
In fact, the normalizing constant satisfies 
\begin{equation}\label{eq:normalization constant}
    \sqrt{n+1}\le C_n \le n+1
\end{equation}
as can be seen from 
\begin{equation*}
    C_n^2 = \sum_{j,k=0}^{n} \langle\Psi^{\ox j}\ox (\lambda\Phi)^{\ox n-j}|\Psi^{\ox k}\ox (\lambda\Phi)^{\ox n-k}\rangle = \sum_{j,k=0}^{n} \langle\Psi|\lambda\Phi\rangle^{|j-k|},
\end{equation*}
since $\ip\Psi{\lambda\Phi}^{|j-k|}$ is in $[0,1]$ for all $j,k$ and equals $1$ if $j=k$.\footnote{The phase $\lambda$ in the definition of $\Omega$ is necessary to ensure the lower bound in \eqref{eq:normalization constant}. While it is possible to absorb it into the vector $\Phi$, this can only be done in a $\Psi$-dependent way, which causes problems later on.}

Next, we show that the sequence $\Omega_n(\Psi,\Phi)$ can be used to embezzle the transition $\Psi\to\Phi$ of $N$-qudit states $\Psi,\Phi$.
We consider the case where each party has $n+1$ qudits.
We denote by $u$ the unitary that implements a cyclic permutation (to the right) on the $n+1$ qudits $(\CC^d)^{\ox n+1}$ and define the $N$-partite local unitary $U = \bar\lambda u^{\otimes N}$ on the Hilbert space $\K_{d,n+1,N}$.
We have (with the obvious implicit reordering of tensor products)
\begin{align} 
    U (\Psi \otimes \Omega_n(\Psi,\Phi)) &= \frac{1}{C_n} U \sum_{j=0}^{n} \Psi\otimes \Psi^{\otimes j} \otimes (\lambda\Phi)^{\otimes n-j} \nonumber\\
    &= \frac{1}{C_n} \bar\lambda\sum_{j=0}^{n-1} \lambda\Phi\otimes \Psi^{\otimes j+1}\otimes (\lambda\Phi)^{\otimes n-j-1} + \frac{\bar\lambda}{C_n} \Psi^{\ox n+1} \nonumber\\
    &= \frac{1}{C_n} \sum_{j=1}^{n} \Phi\otimes\Psi^{\otimes j}\otimes (\lambda\Phi)^{\otimes n-j}+ \frac{\bar\lambda}{C_n} \Psi^{\ox n+1}\nonumber\\
    &= \Phi \otimes \Omega_n(\Psi,\Phi) -\frac{\lambda^n}{C_n} \Phi^{\ox n+1}+\frac{\bar\lambda}{C_n} \Psi^{\ox n+1}.
\end{align}
Thus, \eqref{eq:normalization constant} yields an $N$-independent estimate:
\begin{align}
\label{eq:multipartite-mbz-property}
\norm{U(\Psi\ox\Omega_n(\Psi,\Phi)) - \Phi\ox\Omega_n(\Psi,\Phi)} & \le  \frac{2}{\sqrt{n+1}} < \frac2{\sqrt n}.
\end{align}
In principle, this error could be improved by choosing non-uniform coefficients in the superposition to construct $\Omega_n$; however, we are not concerned with optimality of the construction here \cite{Leung2013coherent}.

We observe the following important property of the construction: Suppose $\Psi$ and $\Phi$ are non-entangled over a bipartition of the $N$ subsystems, which we may take (without loss of generality) as the division into the first $L$ and the remaining $N-L$ subsystems. Assume that $\Phi$ and $\Psi$ coincide on the  last $N-L$ subsystems, i.e.,  $\Psi = \check\Psi\ox \Xi$ and $\Phi = \check\Phi \ox \Xi$, with $\check\Psi,\check\Phi \in (\CC^d)^{\otimes L}$ and $\Xi\in (\CC^d)^{\ox N-L}$ normalized.
We have
\begin{align}
    \Omega_n(\Psi,\Phi)=\Omega_n(\check\Psi\ox\Xi,\check\Phi\ox \Xi) = \Omega_n(\check\Psi,\check\Phi)\ox \Xi^{\otimes n},
\end{align}
with the obvious implicit re-ordering of tensor factors.
Setting $\check U = \bar\lambda u^{\ox L}$, we find
\begin{align}\label{eq:multipartite-local-mbz-property}
   \norm{U (\Psi\ox\Omega_n(\Psi,\Phi))-\Phi\ox\Omega_n(\Psi,\Phi)} & = \norm{\check U(\check\Psi \ox\Omega_n(\check\Psi,\check\Phi))-\check\Phi\ox\Omega_n(\check\Psi,\check\Phi)}<\frac{2}{\sqrt{n}}.
\end{align}
This means that if $L$ of $N$ agents want to transform a local pure state (possibly entangled) on their systems to another local pure state, they can do so without requiring the remaining agents to act on their $N-L$ systems if all the agents share $\Omega_n(\Psi,\Phi)$. 

Let $(\eps_n)_n$ be a decreasing null sequence of positive numbers $\eps_n>0$.
For simplicity, we choose $\varepsilon_n = 1/\sqrt{n}$ in the following. 
For every $n\in\NN$ and $2\leq d\in \NN$ we pick an $\varepsilon_n$-cover $\{\Psi_j^{(d,n)}\}$ of the unit sphere of $(\CC^d)^{\otimes N}$ up to phase: 
For every normalized vector $\Phi \in (\CC^d)^{\otimes N}$ there exists an index $j$ and a phase $\lambda\in S^1$ such that
\begin{align}
    \norm{\lambda \Phi - \Psi_j^{(d,n)}} \leq \eps_n.
\end{align}
Finally, we define
\begin{align}
    \Omega^{(d,n)}:= \bigotimes_{i\neq j} \Omega_n(\Psi_i^{(d,n)},\Psi_j^{(d,n)}) \in \bigotimes_{i\neq j} \K_{d,n,N}\cong \bigotimes_{x=1}^N (\CC^d)^{\otimes nM(d,n)},
\end{align}
where $M(d,n)$ is the number of distinct pairs in the $\eps_n$-cover associated to $(d,n)$.
By construction, $\Omega^{(d,n)}$ can be used to embezzle arbitrary $N$-qudit pure states up to an error of $2/\sqrt{n}+2\eps_n = 4/\sqrt n$.

To turn this into a monopartite embezzling family, we pick a counting $(d_k,n_k)$ of pairs of natural numbers with $d_k\geq 2$ with the property that for every $(d,n)\in\NN\times \NN$ with $d\geq 2$ and every $k_0\in\NN$ there exists a $k\geq k_0$ such that $d=d_k$ and $n_k\geq n$.
Writing $D_k := d_k^{n_k M(d_k,n_k)}$, we regard $\Omega\up{d_k,n_k}$ as a state on the $N$-partite Hilbert space $\bigotimes_{x=1}^N \CC^{D_k}$ and get the Leung-Toner-Watrous $N$-partite embezzling family $(\H\up{n},\Omega\up{n})$ via 
\begin{equation}
    (\H\up{n};\Omega\up{n}) := \bigotimes_{k=1}^n \bigg(\bigotimes_{x=1}^N \CC^{D_k}, \Omega\up{d_k,n_k}\bigg).
\end{equation}

The crucial observation is that we can simply take the limit $n\to \oo$ of the Leung-Toner-Watrous family by using von Neumann's (incomplete) infinite Hilbert space tensor product:
\begin{align}
    (\H,\Omega) := \bigotimes_{k\in\NN} \left(\bigotimes_{x=1}^N \CC^{D_k};\Omega^{(d_k,n_k)}\right).
\end{align}
The Hilbert space $\H$ carries an $N$-partite structure in the sense that the observables of each party $x\in \{1,\ldots,N\}$ are described by a factor $\M_{x}$ acting on $\H$.\footnote{The factor $\M_{x}$ is defined as the weak closure of $\bigcup_{n\in\NN} \left(\bigotimes_{k=1}^n 1_{[1,x-1]} \ox M_{D_k}(\CC) \ox1_{[x+1,N]}\right)$ in its natural action on $\H$.}
By construction, $\M_{x}$ is isomorphic with the ITPFI factor
\begin{align}
\M_{x} \cong \bigotimes_{k\in \NN} (M_{D_k}(\CC); \omega_{x}^{(d_k,n_k)}),
\end{align}
where $\omega_{x}^{(d_k,n_k)}$ is the reduced state of $\Omega^{(d_k,n_k)}$ on the $x$-th tensor factor.

We thus obtain an $N$-partite quantum systems $(\M_{1},\ldots,\M_{N})$ on $\H$ with each subsystem described by a factor $\M_{x}$ and such that $\bigvee_{x=1}^N \M_{x} = \B(\H)$ together with a pure $N$-partite entangled state $\Omega = \otimes_k \Omega^{(d_k,n_k)}$. 
For any $K\in\NN$ we can also consider the increasing sequences of $N$-partite systems $(\M\up{n}_1,\ldots \M\up{n}_N)_n$ via
    \begin{align}
      \M_{x}\up{n} = \bigotimes_{k=1}^n M_{D_k},\quad \H\up{n} = \bigotimes_{k=1}^n (\CC^{D_k})^{\otimes N},\quad  \Omega\up{n} = \bigotimes_{k=1}^n \Omega^{(d_k,n_k)}, 
    \end{align}
and the canonical embedding of $\M\up{n}_x\subset \M\up{n+1}_x$, so that
\begin{align}
    \M_{x} = \M_{x}\up{n} \ox \M_{x}^{>n},\quad \H = \H\up{n}\ox \H^{>n},\quad \Omega = \Omega\up{n}\ox \Omega^{>n}
\end{align}
and $\M_x = (\cup_n \M\up{n}_x)''$. 
The relations \eqref{eq:multipartite-mbz-property} and \eqref{eq:multipartite-local-mbz-property} now turn into the following crucial observation:
By construction of the sequence $\Omega^{(d_k,n_k)}$, for any $d\in\NN$, any normalized $\Psi,\Phi \in (\CC^{d})^{\otimes L}$ with $2\leq L\leq N$, any $\eps>0$ and any $k_0\in \NN$, there exists a $k\geq k_0$ and local unitaries $u_x \in M_{D_k}(\CC)\subset \M_{x}$  for $x=1,\ldots,L$ such that
\begin{align}
    \norm{U (\Phi\ox \Xi) \ox \Omega - (\Psi\ox \Xi) \ox \Omega} \leq \frac{2}{\sqrt n_k} + 2\varepsilon_{n_k} = \frac{4}{\sqrt n_k}<\varepsilon,
\end{align}
where $U=\prod_{x=1}^L u_x$ and where $\Xi\in(\CC^{d})^{\ox N-L}$ is any unit vector. To see this, we simply pick $k$ sufficiently large that $d_k=d$, $4/\sqrt{n_k}<\varepsilon$ and choose the $u_x$ to  perform a cyclic permutation on the relevant tensor factor $\Omega_{n_k}(\Psi_i^{(d_k,n_k)},\Psi_j^{(d_k,n_k)})$ with $\norm{\Psi_i^{(d_k,n_k)}-\Psi\ox\Xi}<1/\sqrt{n_k}$ and $\norm{\Psi_j^{(d_k,n_k)}-\Phi\ox\Xi}<1/\sqrt{n_k}$. The estimate then follows from the triangle inequality.
This yields:

\begin{proposition}
    For any $N\geq 2$ and any $n\in\NN_0$ the $N$-partite system $(\M_{1}^{>n},\ldots,\M_{N}^{>n})$ on $\H^{>n}$ defined above has the following properties:
    \begin{enumerate}
    \item For any subset of $L\leq N$ subsystems with $L\geq 2$, the state $\Omega^{>n}$ is $L$-partite embezzling. 
    \item On every subsystem $\{x\}$ the reduced state $\omega_x^{>n}$ on $\M_{x}^{>K}$ is a monopartite embezzling state.  
 \end{enumerate}
  The family of unit vectors $\Omega\up{n}\in\H\up{n}$ on the increasing sequence of $N$-partite systems $(\M_1\up{n},\ldots, \allowbreak \M_N\up{n})_n$ on $\H\up{n}$ defines a multipartite, consistent embezzling family converging to the embezzling state defined by $\Omega\in\H$. 
\end{proposition}

As a consequence of the proposition, we find that the LTW construction in fact yields a universal multipartite embezzler:
\begin{theorem}
    For any $N\geq 2$, consider the $N$-partite system $(\M_{1},\ldots,\M_{N})$ defined above.
    Then:
    \begin{enumerate}
        \item Each $\M_{x}$ is a hyperfinite type III$_1$ factor.
        \item Every density operator $\rho$ on $\H$ is an $N$-partite embezzling state. 
        \item Haag duality holds.
    \end{enumerate}
\end{theorem}

\begin{proof}
    The first two items follow from the arguments presented in \cite[Sec. 6]{long_paper}. We only sketch them here. For the first item, observe that  $\M_{x}=\M_{x}\up{n}\ox \M_{x}^{>n}$ and that by the above discussion, the states $\omega_x^{>n}$  on $\M_{x}^{>n}$ arising as the marginals of $\Omega^{>n}$ are monopartite embezzling states. Since $\omega_x = \omega_x\up{n}\ox \omega_x^{>n}$ for every $n\in\NN$, it follows from \cite[Thm.~88]{long_paper} that $\M_{x}$ must be of type III$_1$.
    For the second item, every density  operator $\rho$ on $\B(\H)$ can be approximated by a density operator of the form
    \begin{align}
        \rho\up{n} \otimes |\Omega^{>n}\rangle\langle \Omega^{>n}|
    \end{align}
    on $\H\up{n}\ox \H^{>n}$ up to arbitrary accuracy (see for example \cite[Cor. 85]{long_paper}).
    Since $\Omega^{>n}$ defines an $N$-partite embezzling state for every $n\in\NN$, it follows that $\rho$ is an $N$-partite embezzling state.
    Finally, Haag duality holds by \cref{lem:haag}.
\end{proof}

\section{Monopartite and bipartite embezzling families}
\label{sec:families}

Embezzling families were discovered first by van Dam and Hayden in a finite-dimensional bipartite setting \cite{van_dam2003universal}. Bipartite embezzling states and families can equivalently be discussed in a monopartite setting \cite{long_paper}. In this section, we use this to discuss monopartite and bipartite embezzling families and their limits in more detail. 
Our goal is to analyze the existence of limits of two constructions of embezzling families: 
\begin{itemize}
    \item[(1)] \label{it:vdh} The van Dam-Hayden family \cite{van_dam2003universal}.
    \item[(2)] \label{it:ltw} A highly simplified version of the LTW family.
\end{itemize}
For our purposes, we interpret both families as weak*-convergent sequences with respect to a suitable infinite $C^{*}$-tensor product of finite-dimensional matrix algebras. In this way, we obtain a qualitative difference between the two families -- the vDH family leads to the hyperfinite type $\II_{1}$ factor (on which no embezzling states exist) and the simplified LTW family provides us with the hyperfinite type $\III_{1}$ factor -- a universal embezzler. Interestingly, the embezzling unitaries for the simplified LTW family can still be constructed rather explicitly.

We follow \cite{van_luijk_critical_2024} and define a monopartite version of the concept of an embezzling family as follows:
\begin{definition}
\label{def:mbzfam}
    Let $(\M\up{n})_n$ be an increasing net of von Neumann algebras and $(\omega_n)_{n}$ be a net of normal states on $\M\up{n}$.
    Then $(\omega\up{n})_{n}$ is called a \emph{monopartite} embezzling family if for all $d\in\NN$ and all $\eps>0$ there exist an $n=n(d,\eps)$ such that for all $m\ge n(d,\eps)$ and for all states $\psi$ and $\phi$ on $M_{d}(\CC)$ there exist unitaries $u_{m}\in M_{d}(\M\up{m}) = \M\up{m}\ox M_{d}(\CC)$ such that
    \begin{align}
        \norm{u\up{m}(\omega\up{m}\ox\phi)(u\up{m})^{*}-\omega\up{m}\ox\psi} & < \eps.
    \end{align}
\end{definition}
Note that a monopartite embezzling family allows to embezzle arbitrary mixed states on single systems, whereas multipartite embezzling families allow to embezzle arbitrary pure, entangled states. 
For finite-dimensional sequences $(\M\up{n})_{n}$, it is shown in \cite[Prop.~33]{van_luijk_critical_2024} that a monopartite embezzling family induces a bipartite embezzling family in the sense of \cref{def:multipartite-mbz-family} for $N=2$. Conversely, any bipartite embezzling family clearly induces a monopartite embezzling family by passing to its reduced states. Thus, results on monpartite embezzling families induce corresponding results on bipartite embezzling families.  In the following, we therefore restrict to the monopartite setting.

To obtain our results, we first show that it suffices for a monopartite embezzling family to be able to embezzle arbitrary maximally mixed states in order to embezzle arbitrary mixed states. 
There are two ways one can show this:
First, the results in \cite{long_paper} based on the \emph{flow of weights} imply that  for (monopartite) \emph{embezzling states} (instead of families) being able to embezzle maximally mixed states of dimension $d_1$ and $d_2$ such that $\log(d_1)/\log(d_2)$ is irrational is sufficient to be able to embezzle \emph{arbitrary} mixed states. 
Second, by the monopartite analog of \cref{thm:consistent-to-state}, having a \emph{consistent} monopartite family that can embezzle arbitrary maximally mixed states yields a state that can embezzle arbitrary maximally mixed states. Hence, it is a monopartite embezzling state, and its restriction (which yields again the starting family) yields a monopartite embezzling family.

\begin{remark}\label{rem:traces}
    The arguments of \cite{long_paper} also imply that \emph{every} normal state on a type $\III_{1/d}$ factor allows to embezzle maximally mixed states of dimension $d$. Since not all such states are monopartite embezzling states, it follows that being able to embezzle maximally mixed states of a fixed dimension is insufficient to guarantee monopartite embezzlement of all states. In the bipartite setting of embezzlement of entanglement, this means that it is not sufficient to be able to embezzle Bell states via local operations (clearly, it is sufficient if we additionally allow for classical communication \cite{nielsen_conditions_1999}).
\end{remark}

We now present a more direct argument for families on increasing finite-dimensional matrix algebras. It does not require consistency and is based on the following Lemma:

\begin{lemma}\label{lem:contractions}
    Let $\omega$ be a state on $M_k(\CC)$ and $\psi$ a state on $M_d(\CC)$ and suppose that a contraction $s:\CC^k\to \CC^k\ox \CC^d$ fulfills
    \begin{align}
        \norm{\omega\ox \psi - s \omega s^*}<\eps.
    \end{align}
    If $s=va$ is a polar decomposition of $s$ with $0\leq a\leq 1$ and $v:\CC^k\to \CC^k\ox \CC^d$ an isometry, then
      \begin{align}
        \norm{\omega\ox \psi - v\omega v^*}<6\eps^{1/2}.
    \end{align}
\end{lemma}

\begin{proof}
    Every isometry $v:\CC^k\to \CC^k\ox \CC^d$ can be written as $u(1\ox \ket 1)$ for some unitary $u$ on $\CC^k\ox \CC^d$.
    Hence the contraction $r= u(a\ox 1)$ on $\CC^k\ox \CC^d$ satisfies
    \begin{align}
        r(\omega \ox \bra1\placeholder\ket1)r^* = s\omega s^*. 
    \end{align}
    Let $\Omega\in\CC^{k_n}\otimes\CC^{k_n}$ and $\Psi\in\CC^d\ox\CC^d$ be the canonical purifications, so that $\norm{\Omega \ox\Psi - r\ox\bar r(\Omega\ox \ket1\ket1)}< \eps^{1/2}$.
    Let $r\ox\bar r= (u\ox \bar u)b$ be a polar decomposition and put $\Phi=\Omega\ox\ket1\ket1$.
    Then, $\norm{r\ox\bar r \Phi}=\norm{b\Phi} \ge 1-\eps^{1/2}$.
    Since $r\ox\bar r$ is a contraction, we have $0\leq b\leq 1$. Hence $-b \le -b^2$ and
    \begin{equation*}
        \norm{\Phi-b\Phi}^2 = 1 + \bra{\Phi}(b^2-2b)\ket{\Phi}  \le 1- \bra{\Phi}b^2\ket{\Phi} \le 1- (1-\eps^{1/2})^2 = 2\eps^{1/2}-\eps
    \end{equation*}
    Applying the triangle inequality, we obtain 
    \begin{equation*}
        \norm{\Omega\ox\Psi-u\ox\bar u(\Omega\ox\ket1\ket1} )
        \le \eps^{1/2} + 2\eps^{1/2} -\eps \le 3\eps^{1/2}.
    \end{equation*}
    Thus, the standard estimate $\norm{\bra\Phi\placeholder\ket\Phi-\bra\Psi\placeholder\ket\Psi} \le 2 \norm{\Phi-\Psi}$ yields
    \begin{equation*}
     \norm{\omega\ox\psi - v\omega v^*} = \norm{\omega\ox\psi - u(\omega\ox\bra1\placeholder\ket1)u^*} \le 6 \eps^{1/2}.
    \end{equation*}
\end{proof}

\begin{lemma}
\label{lem:mbz_traces}
Let $(\omega\up{n})_{n}$ be a sequence of states on $M_{k_n}(\CC)$ for an increasing sequence of integers $(k_{n})_{n}\subset \NN$. Then, the following are equivalent:
\begin{itemize}
    \item[(1)] $(\omega\up{n})_{n}$ is a monopartite embezzling family in the sense of \cref{def:mbzfam}.
    \item[(2)] $(\omega\up{n})_{n}$ can embezzle maximally mixed states of arbitrary dimension, i.e., for all $d\in\NN$ and $\eps>0$ there exists an $n=n(d,\eps)$ such that for all $m\geq n$ there exist unitaries $u\up{m}\in M_{k_m}(\CC)\ox M_d(\CC)$ such that
    \begin{align}
    \label{eq:mbz_traces}
    \norm{u\up{m}(\omega\up{m}\ox\bra1\placeholder\ket1)(u\up{m})^{*}-\omega\up{m}\ox\tfrac{1}{d}\tr_{d}} & < \eps,
    \end{align}
    where $\tr_{d}$ denotes the canonical trace on $M_{d}(\CC)$.
\end{itemize}
\end{lemma}
\begin{proof}
The implication (1)$\implies$(2) is clear. For the converse statement, (2)$\implies$(1), we argue in analogy with \cite[Cor.~3.3]{haagerup_classification_2009}.

First, if we can embezzle arbitrary mixed states from the initial pure state $\bra1\placeholder\ket1$, then it is possible to embezzle arbitrary mixed states from any other arbitrary mixed state by the triangle inequality and because the inverse of a unitary is unitary. It is, hence, sufficient to show that one can embezzle arbitrary mixed states from $\bra1\placeholder\ket1$.
Second, if $v: \CC^k \to \CC^k\ox \CC^d$ is an isometry, then we can write $v = u(1\ox \ket1)$ for some unitary $u$ on $\CC^k\ox \CC^d$. Conversely, for a unitary $u$, setting $v=u(1\ox \ket1)$ defines an isometry $v$. It hence suffices to consider isometries in the following. We use the shorthand $v\omega v^{*}\approx_{\eps}\omega\ox\psi$ for $\norm{v\omega v^{*}-\omega\ox\psi}  < \eps$.

Consider a state $\psi_{\{\lambda_j\}}$ on $M_{d}(\CC)$ with eigenvalues $\lambda_{1},...,\lambda_{d}\geq0$. 
Without loss of generality, we may assume that $\psi_{\{\lambda_j\}}$ is given in terms of a diagonal density matrix $\rho_{\{\lambda_j\}}$ with respect to the canonical basis of $\CC^{d}$. Therefore, considering $\omega_{n}\ox\psi_{\{\lambda_j\}}$ is equivalent to considering $\omega_{n}\ox\rho_{\{\lambda_j\}}$, and we may write:
\begin{align}
\label{eq:state_ext_prod}
\omega_{n}\ox\rho_{\{\lambda_j\}} & = \begin{pmatrix} \lambda_{1}\omega_{n} & \cdots & 0 \\ \vdots & \ddots & \vdots \\ 0 & \cdots & \lambda_{d}\omega_{n}\end{pmatrix}.
\end{align}
Now, we assume that $\lambda_{j}=\tfrac{p_{j}}{q_{j}}$, $j=1,...,d$, which is sufficient as any state on $M_{d}(\CC)$ can be approximated in norm in this way. We rewrite \cref{eq:state_ext_prod}:
\begin{align}
\omega_{n}\ox\rho_{\{\lambda_j\}} & = \frac{1}{q_{1}...q_{d}}\begin{pmatrix} p_{1}q_{2}...q_{d}\omega_{n} & \cdots & 0 \\ \vdots & \ddots & \vdots \\ 0 & \cdots & q_{1}...q_{d-1}p_{d}\omega_{n}\end{pmatrix}.
\end{align}
Since the family $(\omega_{n})_{n}$ can embezzle traces of arbitrary dimension, we may choose $n$ large enough such that we find an isometry $v_{n,j}:\CC^{k_{n}}\to\CC^{k_{n}}\ox\CC^{q_{1}...p_{j}...q_{d}}$ such that
\begin{align}
v_{n,j}\omega_{n}v_{n,j}^{*} & \approx_{\eps} \omega_{n}\ox\tfrac{1}{q_{1}...p_{j}...q_{d}}\tr_{q_{1}...p_{j}...q_{d}},
\end{align}
for each $j=1,...,d$. Setting $v_{n}=\oplus_{j=1}^{d}v_{n,j}$, we obtain
\begin{align}
v_{n}(\omega_{n}\ox\rho_{\{\lambda_j\}})v_{n}^{*} & \approx_{\eps} \frac{1}{q_{1}...q_{d}}\begin{pmatrix} \omega_{n}\ox\Tr_{p_{1}q_{2}...q_{d}} & \cdots & 0 \\ \vdots & \ddots & \vdots \\ 0 & \cdots & \omega_{n}\ox\Tr_{q_{1}...q_{d-1}p_{d}}\end{pmatrix} = \omega_{n}\ox\tfrac{1}{q_{1}...q_{d}}\tr_{q_{1}...q_{d}}.
\end{align}
Finally, by potentially increasing $n$, we find an isometry $w_{n}:\CC^{k_{n}}\to\CC^{k_{n}}\ox\CC^{q_{1}...q_{d}}$ such that $w_{n}\omega_{n}w_{n}^{*}\approx_{\eps}\omega_{n}\ox\tfrac{1}{q_{1}...q_{d}}\tr_{q_{1}...q_{d}}$. This allows us to define a contraction $s_{n}=v_{n}^{*}w_{n}$ such that:
\begin{align}
s_{n}\omega_{n}s_{n}^{*} & = v_{n}^{*}(w_{n}\omega_{n}w_{n}^{*})v_{n} \approx_{\eps} v_{n}^{*}(\omega_{n}\ox\tfrac{1}{q_{1}...q_{d}}\tr_{q_{1}...q_{d}})v_{n} \approx_{\eps} \omega_{n}\ox\rho_{\{\lambda_j\}},
\end{align}
where we used that $\norm{s\varphi s^{*}}\leq\norm{\varphi}$ for any pair of contraction $s$ and linear functional $\varphi$, and that $v_{n}$ is an isometry. Invoking \cref{lem:contractions} concludes the proof.
\end{proof}

\subsection{The van Dam-Hayden embezzling family}
\label{sec:vdh}

The original embezzling family proposed by van Dam and Hayden is defined as follows:
For each $n\in\NN$, we consider the bipartite pure states $\ket{\Omega\up{n}} = h_{n}^{-\frac{1}{2}}\sum_{j=1}^{n}\frac{1}{\sqrt{j}}\,\ket{jj}$ 
on the Hilbert spaces $\H\up{n} = \CC^{n}\ox\CC^{n}$, where $h_{n} = \sum_{j=1}^{n}\tfrac{1}{j}$ is the $n$th harmonic number.
We denote by 
\begin{align}\label{eq:vDHmbz}
     \omega\up{n} = \frac{1}{h_{n}} \sum_{j=1}^n \frac{1}{j} \bra j\placeholder \ket j
\end{align} 
the marginal of $\Omega\up{n}$ on $M_{n}(\CC)\cong M_{n}(\CC)\ox1\subset M_{n}(\CC)\ox M_{n}(\CC)$. 

\begin{theorem}[\cite{van_dam2003universal}]
\label{thm:vDHmbz}
The sequence of states $(\omega\up{n})_{n}$ in \cref{eq:vDHmbz} is a monopartite embezzling family.
\end{theorem}

We show that a certain weak*-cluster point of the van Dam-Hayden embezzling family, when naturally considered on infinite spin chains, does not yield an embezzling state since the hyperfinite II$_1$ factor does not admit monopartite embezzling states \cite{long_paper}:

\begin{proposition}
\label{prop:vDHmbz}
    The van Dam-Hayden subfamily $\varphi^{(n)}=\omega\up{2^{n}}$, considered as states on the increasing sequence of finite spin chains $\M\up{n}=\bigotimes_{k=1}^{n}M_{2}(\CC)$, converges to the tracial state $\tau$ on $\M_0 = \cup_n \M\up{n}$.
    Hence, we can view $\M_0$ as weakly dense subalgebra of the hyperfinite type II$_1$ factor: $\M = \vee_{n}\M\up{n}$.
\end{proposition}

\begin{proof}
The infinite $C^*$-tensor product $\bigotimes_{k=1}^\oo M_2(\CC) $ is defined as the C*-inductive limit of the inductive system \cite[Sec.~11.4]{KadisonRingrose2}:
\begin{align}
    \M^{(m)}\hookrightarrow \M^{(n)},\qquad a_{m} \mapsto a_{m}\ox1^{\ox(n-m)},\qquad a_{m}\in\M^{(m)}, \ n\geq m.
\end{align}
Then, we have for $a_{m}\geq0$,
\begin{align}
    \varphi\up{n}|_{\M^{(m)}}(a_{m}) & = h_{2^{n}}^{-1}\sum_{l=1}^{2^{n-m}}\sum_{k=1}^{2^{m}}\tfrac{1}{(l-1)2^{m}+k}(a_{m})_{kk},
\end{align}
where we have chosen the orthonormal basis in the definition of $\varphi\up{n}$ to be compatible with the diagonal embeddings, i.e., $a_m\mapsto\mathrm{diag}(a_m,\ldots,a_m)$.\footnote{Note that this is the opposite of the usual convention of the Kronecker product.}
Since $(a_{m})_{kk}\geq0$ and $\tfrac{1}{(l-1)2^{m}+1}\geq\tfrac{1}{(l-1)2^{m}+k}\geq\tfrac{1}{l2^{m}}$, we find:
\begin{align}
    h_{2^{n}}^{-1}h_{2^{n-m}}\tfrac{1}{2^{m}}\tr_{2^{m}}(a_{m}) & \leq \varphi\up{n}|_{\M^{(m)}}(a_{m}) \leq h_{2^{n}}^{-1}(2^{m}+h_{2^{n-m}})\tfrac{1}{2^{m}}\tr_{2^{m}}(a_{m}).
\end{align}
Thus, we find $\lim_{n\rightarrow\infty}\varphi\up{n}|_{\M^{(m)}}(a_{m}) = \tfrac{1}{2^{m}}\tr_{2^{m}}(a_{m}) = \tau|_{\M^{(m)}}(a_{m})$ because $h_{n}=\log(n)+\gamma+\tfrac{1}{2n}-\eps_{n}$ with $0\leq\eps_{n}\leq\tfrac{1}{8n^{2}}$, $\gamma$ the Euler-Mascheroni constant, and $\lim_{n\rightarrow\infty}h_{2^{n}}^{-1}h_{2^{n-m}}=1$.
\end{proof}

\subsection{A simpler Leung-Toner-Watrous embezzling family}
\label{sec:simpler-ltw}
In section~\ref{sec:LTW}, we showed that one can take the limit of the LTW construction to obtain a multipartite, universal embezzler. 
One obvious disadvantage of the construction is its `inefficiency': 
Tensoring $\eps$-covers of high-dimensional unit balls seems too wasteful in terms of Hilbert space dimension. Here, we consider a simpler version of the LTW family in the monopartite setting and show that it yields a universal monopartite embezzler. Via purification, the construction naturally yields a universal embezzler for \emph{bipartite} entanglement.
One way to quantify the efficiency of an embezzling family of states on finite spin chains is in terms of the asymptotic scaling of the circuit complexity. 
This is considered in \cite{schwartzman_complexity_2024}, where the author proves a lower bound on the complexity of all embezzling families.

In this subsection, we do two things: 
We note a simpler version of the LTW family, which is defined in the same way as the LTW family but only takes maximally entangled states. 
Moreover, we give direct proof of the claim that the limiting system is of type $\III_1$ by directly computing the \emph{asymptotic ratio set} $r_\infty(\M)$ of the resulting von Neumann algebra.

Given two states $\psi,\varphi$ on a $d$-dimensional quantum system, i.e., on $M_{d}(\CC)$, we can construct an approximate (within an error $\tfrac{1}{n-1}$) catalyst $\omega^{(n)}_{d}$ that facilitates the transition from $\psi$ to $\varphi$ by
\begin{align}
\label{eq:approx_cat}
\omega^{(n)}_{d} & = \tfrac{1}{n-1}\sum_{k=1}^{n-1}\psi^{\ox k}\ox\varphi^{\ox n-k},    
\end{align}
which satisfies
\begin{align}
\label{eq:approx_cat_err}
\norm{u(\psi\ox\omega^{(n)}_{d})u^{*}-\varphi\ox\omega^{(n)}_{d}} & \leq \tfrac{2}{n-1},
\end{align}
with $u\in M_{d}(\CC)^{\ox n}$ being the unitary that cyclically permutes the tensor factors.\footnote{Note that we use approximate catalysts of a slightly different form compared to \cref{sec:LTW}.}

Let us specialize to the state $\omega^{(n)}_{d}$ for $\psi = \bra{1}\placeholder\ket{1}$ and $\varphi = \tfrac{1}{d}\tr_{d}$. Due to \cref{lem:mbz_traces}, we can form an embezzling family from the states $\{\omega^{(n)}_{d}\}$, for all dimensions $d\geq 2$ and error-parameters $n\geq 2$, by considering the product states
\begin{align}
\label{eq:ltw_fam}
\omega^{(N,D)} & = \ox_{n=2}^{N}\ox_{d=2}^{D}\omega^{(n)}_{d}
\end{align}
on $\M^{(N,D)}=\ox_{n=2}^{N}\ox_{d=2}^{D}M_{d}(\CC)^{\ox n}$.
We can take the limit $N,D\to \oo$ of $\M\up{N,D}$ using the incomplete infinite tensor product 
\begin{equation}\label{eq:monopartite_LTW_limit}
    \M = \bigotimes_{n,d=2}^\oo (M_d(\CC)^{\ox n}; \omega_d\up n).
\end{equation}
To understand the properties of this algebra, we need to analyze the spectrum of the states $\omega^{(n)}_{d}$:

\begin{lemma}
\label{lem:ltw_spec}
Consider the state $\omega^{(n)}_{d}$ on $M_{d}^{\ox n}$ for $\psi = \bra{1}\placeholder\ket{1}$ and $\varphi = \tfrac{1}{d}\tr_{d}$. The spectrum of $\omega^{(n)}_{d}$ has the following form:
\begin{align}
\label{eq:ltw_spec}
\lambda^{(n,d)}_{j} & = \frac{1}{n-1}\frac{d}{d-1}\Big(\frac{1}{d^{j}}-\frac{1}{d^{n}}\Big), & m^{(n,d)}_{j} & = (d-1)d^{j-1}+\delta_{j,1}, & j & = 1,...,n,
\end{align}
where $\{\lambda^{(n,d)}_{j}\}$ denote the spectral values and $\{m^{(n,d)}_{j}\}$ are their multiplicities.
\end{lemma}
\begin{proof}
This follows from a direct computation. To this end, we also note that $\psi = \bra{1}\placeholder\ket{1}$ and $\varphi = \tfrac{1}{d}\tr_{d}$ are simultaneously diagonal.
\end{proof}

We may use \cref{lem:ltw_spec} to prove that the limit algebra $\M$ (see \eqref{eq:monopartite_LTW_limit}) is the hyperfinite type $\III_{1}$ factor. 
To this end, we note that $\M$ is an ITPFI factor, and we compute the asymptotic ratio set $r_{\infty}(\M)$ as introduced by Araki and Woods \cite{araki1968factors}. 
We then use the fact that $\M$ is the hyperfinite type III$_1$ if and only if $r_\oo(\M)=\RR_+$ \cite{haagerup_uniqueness_1987,connes1973classIII}.

\begin{definition}{\cite[Def.~3.2]{araki1968factors}}
\label{def:asymptotic_ratio}
Let $I\subset\{2,...,\infty\}^{\times 2}$ be a finite subset. Then $\sigma_{I}(\M,\omega)$ denotes the non-zero part of the spectrum $\ox_{(n,d)\in I}\omega^{(n)}_{d}$ on $\ox_{(n,d)\in I}M_{d}(\mathbb C)^{\ox n}$. For a subset $K\subset\sigma_{I}(\M,\omega)$, we define
\begin{align}
\label{eq:spec_subsum}
\lambda(K) & = \sum_{\lambda\in K}\lambda.
\end{align}
The asymptotic ratio set $r_{\infty}(\M)\subseteq\RR_{+}$ consists of all $x\geq0$ such that there exists a sequence of mutually disjoint subsets $I_{n}\subset\{2,...,\infty\}^{\times 2}$, mutually disjoint subsets $K^{1}_{n},K^{2}_{n}\subset\sigma_{I_{n}}(\M,\omega)$, and bijections $\beta_{n}:K^{1}_{n}\to K^{2}_{n}$ such that
\begin{align}
\label{eq:spec_subsum_infty}
\sum_{n}\lambda(K^{1}_{n}) & = \infty, & \lim_{n}\max_{\lambda\in K^{1}_{n}}|x-\beta_{n}(\lambda)/\lambda| & = 0.
\end{align}
\end{definition}
\begin{lemma}
\label{lem:ltw_ratios}
Let $\M$ be the factor arising from the limit of the LTW family as in \eqref{eq:monopartite_LTW_limit}.
Then, we have $r_{\infty}(\M) = \RR_{+}$.
\end{lemma}
\begin{proof}
For any $d_{1},d_{2}$, we consider the mutually disjoint index sets $I_{n}=\{(n,d_{1}),(n,d_{2})\}$ for $n\geq2$. We choose size-$\min\{d_{1},d_{2}\}$ subsets $K^{j}_{n}=\{\lambda^{(n,d_{j})}_{1},...,\lambda^{(n,d_{j})}_{1}\}$, $j=1,2$, of $\sigma_{I_{n}}(\M,\omega)$ according to \cref{lem:ltw_spec}. We define the bijections $\beta:K^{1}_{n}\to K^{2}_{n}$ by
\begin{align}
\label{eq:ltw_bij}
\beta_{n}\Big(\lambda^{(n,d_{1})}_{1}\Big) & = \lambda^{(n,d_{2})}_{1}.
\end{align}
Then, we have
\begin{align}
\label{eq:ltw_subsum}
\sum_{n=2}^{\infty}\lambda(K^{1}_{n}) & = \frac{d_{1}}{d_{1}-1}\frac{\min\{d_{1},d_{2}\}}{d_{1}}\sum_{n=2}^{\infty}\frac{1}{n-1}\Big(1-d_{1}^{-(n-1)}\Big) \geq \frac{\min\{d_{1},d_{2}\}}{d_{1}}\sum_{n=1}^{\infty}\frac{1}{n} = \infty
\end{align}
and
\begin{align}
\label{eq:ltw_ratios}
\beta_{n}\Big(\lambda^{(n,d_{1})}_{1}\Big)/\lambda^{(n,d_{1})}_{1} & = \lambda^{(n,d_{2})}_{1}/\lambda^{(n,d_{1})}_{1} = \frac{d_{1}-1}{d_{2}-1}\left(\frac{1-d_{2}^{-(n-1)}}{1-d_{1}^{-(n-1)}}\right).
\end{align}
Thus, we find that $x = \tfrac{d_{1}-1}{d_{2}-1}\in r_{\infty}(\M)$. This implies that $r_{\infty}(\M)=\RR_{+}$ because the asymptotic ratio set is closed \cite{araki1968factors}.
\end{proof}

As an immediate consequence, we obtain the characterization of the limit of the simplified LTW family.

\begin{theorem}\label{thm:ltw_type}
    Let $\M$ be the factor arising from the limit of the LTW family as in \eqref{eq:monopartite_LTW_limit}.
    Then $\M$ is the hyperfinite type $\III_1$ factor. In particular, it is a universal monopartite embezzler.
\end{theorem}

\appendix

\section{Split bipartite systems admit no embezzling states}

In this appendix, we prove a no-go theorem for bipartite embezzlement in split bipartite systems.
In fact, we will rule out the possibility of embezzling a single pure entangled state.

Let $\Psi\in\CC^d\ox\CC^d$ and let $\omega$ be a bipartite state on the bipartite system $(\M_1,\M_2)$, i.e., $\omega$ is a normal state on $\M_1\vee\M_2$.
We say that \emph{$\Psi$ can be embezzled from $\omega$} (or that \emph{$\omega$ is $\Psi$-embezzling}) if, for every $\eps>0$, there exist unitaries $u_1\in M_d(\M_1)$ and $u_2\in M_d(\M_2)$ such that
\begin{equation}
    \norm{u_1u_2 (\omega\ox \bra{11}\placeholder\ket{11})u_1^*u_2^*-\omega \ox \bra\Psi\placeholder\ket\Psi}<\eps.
\end{equation}
Similarly, we say that a state $\psi$ on $M_d(\CC)$ can be (monopartitely) embezzled from a normal state $\omega$ a von Neumann algebra $\M$ if, for every $\eps>0$, there is a unitary $u\in M_d(\M)$ such that $\norm{u(\omega\ox\bra1\placeholder\ket1)u^*-\omega\ox\psi}<\eps$.

We show that in a split bipartite system, no normal state on $\M_1\vee\M_2$ is able to embezzle even a single entangled state vector $\Psi\in\CC^d\ox\CC^d$.

\begin{theorem}\label{thm:no-go}
    Let $(\M_1,\M_2)$ be a split bipartite system. Then, given an entangled state vector $\Psi\in\CC^d\ox\CC^d$, there does not exist a $\Psi$-embezzling, normal state on $\M_1\vee\M_2$.
\end{theorem}

As a byproduct of our proof, we obtain that any system which admits embezzlement of a single entangled state must be type III.

\begin{proposition}\label{prop:typeIII}
Let $(\M_1,\M_2)$ be a bipartite system together with a normal state $\omega$ on $\M_1\vee\M_2$, and let $\Psi\in\CC^d\ox\CC^d$ be an entangled state vector. For $j=1,2$, let $z_{j}\in\M_j$ be the central support projection of the marginal $\omega_{j}$.
Then, if $\omega$ is $\Psi$-embezzling, $z_{j}\M_j$ must be a type $\III$ von Neumann algebra.    
\end{proposition}

Recall that if $\M\subset \B(\H)$ is a standard representation, then each normal state $\omega$ has a unique representative $\Omega_\omega$ in the associated positive cone $\P$. 
Moreover, we have:
\begin{align}\label{eq:F-vdG}
\norm{\Omega_\omega - \Omega_\psi}^2 \leq \norm{\omega-\psi} \leq \norm{\Omega_\omega+\Omega_\psi}\norm{\Omega_\omega - \Omega_\psi}.
\end{align}
The standard representation of $\M\ox M_d(\CC) = M_d(\M)$ has Hilbert $\H \ox \CC^d \ox \CC^d$ and modular conjugation $J\up d( \sum \Phi_{ij}\ox \ket{ij} ) = \sum J\Phi_{ij} \ox \ket{ji}$ (see \cite[Lem.~2]{long_paper} which also contains a description of the positive cone) and the representation of $x=[x_{k,l}]\in M_d(\M)$ takes the form
\begin{align}
 x=   \sum_{k,l} x_{k,l} \ox |k\rangle\langle l|\ox 1,\qquad x_{k,l}\in \M,
\end{align}
and its modular conjugate is given by
\begin{align}
    j\up d(x) = \sum_{k,l}j_i(x_{k,l})\ox 1\ox |k\rangle\langle l|. 
\end{align}

\begin{lemma}\label{lem:no-periods}
    Denote by $\theta_x$, $x\in\RR$, the action of $\RR$ on itself by translations.
    Let $1> p_1\ge ... \ge p_d \ge 0$ with $\sum_{i=1}^d p_i = 1$ and let $P$ be a Borel probability measure on $\RR$. Then
    \begin{align}\label{eq:no-periods}
        P \neq \sum_{i=1}^d p_i\, \theta_{\log p_i}(P),
    \end{align}
    where $\theta_t(P)$ is the measure defined by $\theta_t(P)(A) = P(\theta_{-t}(A))$.
\end{lemma}
\begin{proof}
    For $\eps>0$, let $n>0$ be such that $P([-n,n])\ge 1-\eps$.
    Let $N>0$ be such that $N \cdot \min |\log p_i| > 2n$.
    For a multi-index $\alpha=(\alpha_1,...,\alpha_{N}) \in [d]^N$, set $p^\alpha = \prod_{i=1}^{N} p_{\alpha_i}$.
    For the sake of contradiction, assume that \eqref{eq:no-periods} holds with equality. Then we can iterate the assumed equality $N$ times to get
    \begin{equation*}
        P = \sum_{\alpha_1,\ldots\alpha_N=1}^d p^\alpha  \, \theta_{\log p^\alpha}(P).
    \end{equation*}
    By assumption, $\theta_{-\log p^{\alpha}}([-n,n])$ is disjoint from $[-n,n]$.
    Hence $P(\theta_{-\log p^\alpha}([-n,n]))\le\eps$ and we get
    \begin{align*}
        1-\eps \le P([-n,n]) &= \sum_{\alpha_1,\ldots\alpha_N=1}^d p^\alpha P(\theta_{-\log p^\alpha}([-n,n])) 
        \le \sum_{\alpha_1,\ldots\alpha_N=1}^d p^\alpha \eps =\eps,
    \end{align*}
    which is a contradiction for sufficiently small $\eps>0$.
\end{proof}
\begin{corollary}\label{cor:typeIII}
    Let $\M$ be  a von Neumann algebra and let $\omega$ be a state on $\M$ with central support projection $z$.
    If a single mixed state can be embezzled from $\omega$, then $z\M$ is a type III von Neumann algebra.
\end{corollary}

\begin{proof}
    We may assume that $\omega$ has central support $z=1$.
    Following \cite{long_paper}, we consider the flow of weights $(\A,\theta)$ of $\M$, which is an abelian von Neumann algebra $\A$ together with a point-ultraweakly continuous flow $\theta :\RR\curvearrowright\A$.
    By \cite[Cor.~47]{long_paper}, the assumption that $\omega$ is $\psi$-embezzling implies that the spectral state $\hat \omega$, which is a normal state on $\A$, satisfies the identity
    \begin{equation}\label{eq:convolution}
        \sum p_i \,\hat\omega\circ\theta_{\log p_i} = \hat\omega,
    \end{equation}
    where $1\ge p_1\ge p_2\ldots \ge p_d\ge 0$ are the eigenvalues of $\psi$ (repeated according to their multiplicity).
    We start by showing the claim under the assumption that $\M$ is a factor.
    If $\M$ were semifinite, the flow of weights would be isomorphic to translations on $\RR$ and, by \cref{lem:no-periods}, this gives $p_1=1$ and $p_2,..., p_d =0$. However, since $\Psi$ is entangled, this means that $\M$ cannot be semifinite and, hence, has type III.

    We now come to the case where $\M$ is a general von Neumann algebra. 
    Consider the disintegration $\M=\int_X^\oplus \M\up x\,d\mu(x)$ into a direct integral of factors $\M\up x$. 
    In this case, the flow of weights is given by a direct integral $(\A,\theta)=\int^\oplus_X (\A\up x, \theta\up x)\,d\mu(x)$, where $(\A\up x,\theta\up x)$ is the flow of weights of $\M\up x$ \cite[Prop.~8.1]{haagerup1990equivalence}.
    If $\omega = \int^\oplus_X p_x\omega_x \,d\mu(x)$ is the disintegration of $\omega$, where $\omega_x$ is a measurable field of states and $p_x\in L^1(X,\mu)$ is a faithful probability density, then the spectral state $\hat \omega$ decomposes as $\hat \omega = \int^\oplus_X p_x \,\hat\omega_x\circ \theta_{\pm \log p_x}\up x \,d\mu(x)$.
    Therefore, \eqref{eq:convolution} implies that
    \begin{equation}
        \sum p_i\, \hat\omega_x \circ\theta_{\log p_i}\up x = \hat\omega_x
    \end{equation}
    holds $\mu$-almost everywhere.
    As before, this is impossible unless $\mu$-almost all $\M\up x$ are type $\III$.
\end{proof}
\begin{proof}[Proof of \cref{prop:typeIII}]
    If $\omega$ is $\Psi$-embezzling, then $\omega_j$ is (monopartitely) $\psi$-embezzling where $\psi$ is the marginal state induced by $\Psi\in \CC^d\ox \CC^d$ on $M_d(\CC)$.
    Since $\psi$ is mixed if and only if $\Psi$ is entangled, the result follows from \cref{cor:typeIII}.
\end{proof}

\begin{proof}[Proof of \cref{thm:no-go}]
    We may assume $\M_1\vee\M_2=\M_1\ox\M_2$.
    Let $(\H_j,\P_j,J_j)$ be the standard form of $\M_j$, $j=1,2$ \cite{haagerup_standard_1975}. 
    Then the induced representation $\M_1\ox \M_2\subset\B(\H_1\ox\H_2)$ is also standard with $J = J_1\ox J_2$ \cite[Prop.~8.1]{stratila2020modular}.
    Let $\Omega_\omega\in \H_1\ox\H_2$ be the implementing vector of the state $\omega$ on $\M_1\ox\M_2$.
    The standard representation of $\M_1\ox \M_2 \ox M_d(\CC)\ox M_d(\CC)$ is given on the Hilbert space $\H_1\ox \H_2 \ox (\CC^d)^{\ox 4}$ with modular conjugation $\tilde j =  j_1\up d\ox j_2\up d$ (under appropriate re-ordering of tensor factors).
    The implementing vector of the state $\omega\ox \bra\Psi\placeholder\ket\Psi$ on $\M_1\ox \M_2\ox M_d(\CC)^{\ox 2}$ is
    \begin{equation}
        \Omega_\omega \ox \Psi\ox \bar\Psi \in \H_1\ox \H_2\ox (\CC^d)^{\ox4},
    \end{equation}
    where $\bar\Psi = \sum \bar \Psi_{ij} \ket{ij}$ if $\Psi=\sum \Psi_{ij}\ket{ij}$.
    Note that the vector $\Psi\ox\bar\Psi \in (\CC^d)^{\ox4}$ is entangled for the bipartite system corresponding to the parties $1$ and $2$.

    Next, we will argue that the assumptions imply that the vector $\Omega_\omega$ is a $(\Psi\ox\bar\Psi)$-embezzling bipartite pure state on the bipartite system $(\B(\H_1)\ox1,1\ox\B(\H_2))$.
    If unitaries $(u_1,u_2)\in M_d(\M_1)\times M_d(\M_2)$ are chosen such that $(u_1\ox u_2)(\omega\ox \bra{11}\placeholder\ket{11})(u_1^*\ox u_2^*)\approx \omega\ox \bra\Psi\placeholder\ket\Psi$, then \eqref{eq:F-vdG} gives
    \begin{equation}
        (u_1 u_1'\ox u_2 u_2') (\Omega_\omega\ox \ket{1111}) \approx \Omega_\omega\ox\Psi \ox \bar\Psi,
    \end{equation}
    where  $u_i' = j\up d(u_i) \in M_d(\M_i)'$.
    Setting $v_A= u_1u_1' \in \B(\H_1)\ox M_d(\CC)^{\ox 2}$ and $v_B = u_2u_2'\in \B(\H_2)\ox M_d(\CC)^{\ox 2}$, we obtain local unitaries that embezzle the entangled state $\Psi \ox \bar \Psi$ from the bipartite state $\Omega_\omega$ on the bipartite system $(\B(\H_1)\ox1,1\ox\B(\H_2))$.
    This is a contradiction with \cref{prop:typeIII}.
\end{proof}

\phantomsection
\addcontentsline{toc}{section}{References}

\fussy
\emergencystretch=1.4em
\printbibliography

\end{document}

%% file: common_preamble_quantum.tex
\usepackage[utf8]{inputenc}
\usepackage[T1]{fontenc}
\usepackage[english]{babel}
\usepackage{amsmath,amssymb,amsthm,mathdots}
\usepackage{graphicx} 
\usepackage[margin=1.in]{geometry}
\usepackage[shortlabels]{enumitem}
\usepackage{csquotes}\MakeOuterQuote"
\usepackage{soul}\setstcolor{red}
\usepackage{color}
\usepackage[colorlinks,allcolors=quantumviolet]{hyperref}
\usepackage[capitalize]{cleveref}
\usepackage{comment}
\usepackage[bibstyle=numeric, backend=biber, sorting=none, citestyle=numeric-comp, giveninits,url=false,maxbibnames=5]{biblatex}
\usepackage{orcidlink}
\newtheorem{theorem}{Theorem}
\newtheorem{lemma}[theorem]{Lemma}
\newtheorem{corollary}[theorem]{Corollary}
\newtheorem{proposition}[theorem]{Proposition}

\newtheorem{definition}[theorem]{Definition}
\newtheorem{introtheorem}{Theorem} 

\newtheorem{introlemma}[introtheorem]{Lemma}

\newtheorem*{introquestion}{Open Question}
\newtheorem{introproposition}[introtheorem]{Proposition}

\theoremstyle{definition}
\newtheorem*{definition*}{Definition}

\newtheorem*{problem*}{Problem}
\newtheorem*{assumption*}{Assumption}

\newtheorem{remark}[theorem]{Remark}
\newtheorem*{warning*}{Warning}

\newcommand{\ip}[2]{\langle #1|#2\rangle}
\newcommand{\braket}[2]{\langle #1|#2\rangle}

\newcommand{\ket}[1]{|#1\rangle}

\newcommand{\bra}[1]{\langle#1|}

\newcommand{\norm}[1]{\lVert #1\rVert}
\newcommand{\oo}{\infty}
\newcommand{\ox}{\otimes}
\newcommand{\mc}{\mathcal}
\newcommand{\eps}{\varepsilon}
\newcommand{\III}{{\mathrm{III}}}
\newcommand{\II}{{\mathrm{II}}}

\newcommand{\up}[1]{^{(#1)}}

\DeclareMathOperator{\tr}{Tr}
\DeclareMathOperator{\Tr}{Tr} 

\renewcommand{\tilde}{\widetilde}
\renewcommand{\hat}{\widehat}

\newcommand{\hide}[1]{}

\def\A{{\mc A}}

\def\B{{\mc B}}
\def\CC{{\mathbb C}}

\def\H{{\mc H}}
\def\K{{\mathcal K}}
\def\M{{\mc M}}
\def\N{{\mc N}}

\def\U{{\mc U}}

\def\RR{{\mathbb R}}

\def\NN{{\mathbb N}}
\renewcommand\P{\mc P}

\newcommand{\placeholder}[0]{{\,\cdot\,}}



%% file: figure.tikz.tex
\begin{tikzpicture}
	
	\draw (0.2,0.25) node[right]{$\varphi^{(1)}_{\phantom]}$} (1.75,0.25) node[right]{$\varphi^{(2)}_{1\phantom]}$} (4.75,0.25) node[right]{$\varphi^{(n)}_{1\phantom]}$} (7.75,0.25) node[right]{$\varphi^{(n+k)}_{1\phantom]}$}	(10.75,0.25) node[right]{$\omega\up{1}_{\phantom]}$} (12.5, 0.25) node[right]{$\M\up{1}_{[N]}$};
	\draw (3.25, 0.25) node[right]{$\dots$};
    \draw (6.5, 0.25) node[right]{$\dots$};
    \draw (9, 0.25) node[right]{$\dots$};
	
	\draw (1.75,1.5) node[right]{$\varphi^{(2)}_{\phantom]}$} (4.75,1.5) node[right]{$\varphi^{(n)}_{2\phantom]}$} (7.75,1.5) node[right]{$\varphi^{(n+k)}_{2\phantom]}$} (10.75,1.5)	node[right]{$\omega\up{2}_{\phantom]}$}  (12.5, 1.5) node[right]{$\M\up{2}_{[N]}$};
	\draw (3.25, 1.5) node[right]{$\dots$};
	\draw (6.5, 1.5) node[right]{$\dots$};
    \draw (9, 1.5) node[right]{$\dots$};
	
	\draw[->] (2,1.1) to (2,0.6);
	\draw[->] (4.975,1.1) to (4.975,0.6);
    \draw[->] (7.975,1.1) to (7.975,0.6);
     \draw[->] (11,1.1) to (11,.6);
    \draw (12.95, 0.875) node{$\cup$};
	
	\draw (4.75,3.5) node[right]{$\varphi^{(n)}_{\phantom ]}$} (7.75,3.5) node[right]{$\varphi^{(n+k)}_{n\phantom]}$} (10.75,3.5) node[right]{$\omega\up{n}_{\phantom]}$} (12.5, 3.5) node[right]{$\M\up{n}_{[N]}$};
    \draw (6.5, 3.5) node[right]{$\dots$};
    \draw (9, 3.5) node[right]{$\dots$};
	
	\draw[->] (4.975,3.125) to (4.975,2.625);
    \draw[->] (7.975,3.125) to (7.975,2.625);
     \draw[->] (11,3.125) to (11,2.625);
    \draw (12.95, 2.925) node{$\cup$};

    \draw (7.75,5.5) node[right]{$\varphi^{(n+k)}_{\phantom{]}}$} (10.75,5.5) node[right]{$\omega\up{n+k}_{\phantom{]}}$} (12.5, 5.5) node[right]{$\M\up{n+k}_{[N]}$};
    \draw (9, 5.5) node[right]{$\dots$};

    \draw[->] (7.975,5.125) to (7.975,4.625);
    \draw[->] (11,5.125) to (11,4.625);
    \draw (12.95, 4.925) node{$\cup$};
	
	\node[right] at (3.25,2.625) {$\iddots$};
	\node[right] at (4.825,2.25) {$\vdots$};
	\node[right] at (7.8,2.25) {$\vdots$};
    \node[right] at (10.8,2.25) {$\vdots$};
    \node[right] at (12.775,2.25) {$\vdots$};

    \node[right] at (6.5,4.825) {$\iddots$};

    \node[right] at (3.25,2.625) {$\iddots$};
	\node[right] at (7.8,4.25) {$\vdots$};
    \node[right] at (10.8,4.25) {$\vdots$};
    \node[right] at (12.775,4.25) {$\vdots$};

    \node[right] at (8.8,6.25) {$\iddots$};
    \node[right] at (10.8,6.25) {$\vdots$};
    \node[right] at (12.775,6.25) {$\vdots$};
    
\end{tikzpicture}